\PassOptionsToPackage{margin=2cm}{geometry}
\documentclass[sigconf, nonacm]{acmart}

\usepackage{soul}
\usepackage{flushend}

\usepackage{pgfplots}
\pgfplotsset{compat=1.18}
\usepackage{caption}
\usepackage{subcaption}
\usepackage{microtype}
\usepackage[nomath]{lmodern}
\usepackage[T1]{fontenc}

\definecolor{listdbcolor}{HTML}{1f77b4}
\definecolor{veliriscolor}{HTML}{d62728}
\definecolor{groupcommitcolor}{HTML}{2ca02c}
\pgfplotsset{
  bench style/.style={
    width=\linewidth,
    height=4.5cm,
    grid=both,
    grid style={line width=0.2pt, draw=gray!25},
    major grid style={line width=0.5pt, draw=gray!55},
    minor tick num=1,
    tick align=outside,
    legend style={
      font=\scriptsize,
      at={(0.03,0.97)},
      anchor=north west,
      fill=white, fill opacity=0.85,
      draw=gray!50,
      text opacity=1,
      row sep=0.5pt,
    },
    label style={font=\small},
    tick label style={font=\footnotesize},
    title style={font=\small\bfseries, yshift=2pt},
    every axis plot/.append style={line width=1.5pt},
    mark size=2.8pt,
    enlarge x limits=0.05,
    enlarge y limits=0.10,
  },
}

\everypar{\looseness=-1}

\usepackage{wrapfig}

\usepackage{dokeeffe-macros}
\usepackage{algorithm}
\usepackage{algpseudocode}

\usepackage{xspace}
\usepackage{graphicx} 
\usepackage{hyperref}

\makeatletter
\pretocmd{\algorithmic}{%
	  \renewcommand{\theHALG@line}{\thealgorithm.\arabic{ALG@line}}%
		}{}{}
\makeatother

\usepackage{listings}
\usepackage[inline,shortlabels]{enumitem}

\usepackage{pgfplots}
\usepgfplotslibrary{groupplots}

\usepackage{amsmath}
\usepackage{xcolor}

\usepackage{dblfloatfix}
\usepackage{placeins}

\usepackage{tabularx} 
\usepackage{booktabs} 
\usepackage{makecell} 

\usepackage{pifont}
\colorlet{dgreen}{green!80!black}

\newcommand{\cmark}{\textcolor{dgreen}{\ding{51}}} 
\newcommand{\xmark}{\textcolor{red}{\ding{55}}}   

\usepackage{mfirstuc}

\usepackage{pgfplots}

\usepackage{circuitikz}
\usepackage{tikz}
\usetikzlibrary{positioning,shapes.geometric,arrows.meta,calc}

\usepackage[capitalise]{cleveref}

\newcommand{\PutC}{\Op{put}\xspace}
\newcommand{\GetC}{\Op{get}\xspace}
\newcommand{\DeleteC}{\Op{delete}\xspace}

\newcommand{\WriteBatchC}{\Op{write\_batch}\xspace}

\newcommand{\phaseLocate}{\emph{locate}\xspace}
\newcommand{\phasePrepare}{\emph{prepare}\xspace}
\newcommand{\phaseAttach}{\emph{attach}\xspace}
\newcommand{\phasePromote}{\emph{promote}\xspace}

\crefname{lstlisting}{Listing}{Listings}
\Crefname{lstlisting}{Listing}{Listings}

\definecolor{codegreen}{rgb}{0,0.6,0}
\definecolor{codegray}{rgb}{0.5,0.5,0.5}
\definecolor{codepurple}{rgb}{0.58,0,0.82}
\definecolor{backcolour}{rgb}{0.95,0.95,0.92}

\lstdefinestyle{mystyle}{
    language=C,
    commentstyle=\color{codegreen},
    keywordstyle=\color{magenta},
    numberstyle=\tiny\color{codegray},
    stringstyle=\color{codepurple},
    basicstyle=\ttfamily\footnotesize,
    breakatwhitespace=false,         
    breaklines=true,                 
    captionpos=b,     
    xleftmargin=10pt,
    keepspaces=true,                 
    numbers=left,                    
    numbersep=5pt,                  
    showspaces=false,                
    showstringspaces=false,
    showtabs=false,                  
    tabsize=2,
    keywords={typedef, struct, char, void, end, let, each, for, if, await, then, else, while, do, in, return, global, bool, continue, this, typename, template},
    morekeywords = [2]{initial, true, false, state, ghost, external, thread, special, actions, labels, id, variables, automaton, specification, using, sfence,mfence,flush_opt,STORE, assertions, rely, condition, declarations, with, LOAD, COMPARE_EXCHANGE_STRONG, FLUSH, CRASH, FENCE, compare_exchange_strong},
    keywordstyle= [2]\color{purple}
    }

\newcommand{\key}{\ensuremath{k}}
\newcommand{\keys}{\mathsf{keys}}
\newcommand{\allkeys}{\mathcal{K}}
\newcommand{\valu}{\ensuremath{v}}
\newcommand{\valus}{\mathsf{values}}

\newcommand{\oread}{\mathsf{read}}
\newcommand{\owrite}{\mathsf{write}}
\newcommand{\rt}{\mathsf{rt}}

\newcommand{\allrdr}{\mathsf{Reads}}
\newcommand{\allwrr}{\mathsf{Writes}}
\newcommand{\wrr}{\ensuremath{w}}
\newcommand{\rdr}{\ensuremath{r}}

\newcommand{\WR}{\mathsf{wr}}
\newcommand{\WW}{\mathsf{ww}}
\newcommand{\RW}{\mathsf{rw}}
\newcommand{\RT}{\mathsf{rt}}
\newcommand{\vis}{\mathsf{vis}}

\newcommand{\sys}{{\sc FlintKV}\xspace}
\DeclareRobustCommand{\sysrob}{{\sc FlintKV}\xspace}
\newcommand{\sysplain}{FlintKV}

\crefformat{section}{\S#2#1#3}

\title{\sys: A Fast Durable Storage Engine for Modern Databases}

\author{Sergey Egorov$^{1}$$^{2}$, Gregory Chockler$^{2}$, Brijesh Dongol$^{2}$}
\author{Dan O'Keeffe$^{1}$, Sadegh Keshavarzi$^{2}$}

\affiliation{%
  $^{1}$Royal Holloway, University of London, Egham, UK\ \ 
  $^{2}$University of Surrey, Guildford, UK
}

\begin{document}
\begin{abstract}
Byte-addressable non-volatile memory (NVM) offers an opportunity to rethink storage engine architectures. 
While recent \NVMM key-value stores achieve high throughput for ingestion and point lookups, they omit or underspecify the support for the richer interface guarantees required by modern databases. 
Production key-value engines (e.g., RocksDB) provide point-in-time snapshots, consistent iterators, and atomic batches—features essential for implementing transactions and concurrency control.

We present \sys, an NVM-optimised skiplist-based storage engine that natively supports the full API of production key-value stores. 
\sys supports both atomic batch writes and snapshot-consistent iteration
	efficiently while guaranteeing durable linearizability. \sys can be deployed
	standalone or its durable skiplist can be integrated into existing NVM stores
	to enhance their capabilities. Central to \sys is a novel flat-combining-based concurrency control algorithm that leverages multi-versioning and
	carefully co-designed persistence mechanisms to ensure high performance and
	scalability. Our empirical evaluation shows that \sys can achieve up to a
	75\% improvement in end-to-end throughput over prior work.
\end{abstract}

\maketitle

\bibliographystyle{ACM-Reference-Format}

\section{Introduction}

Modern transactional database systems increasingly rely on persistent
key-value (KV) stores as their core storage
engines~\cite{cockroachdb,tidb,myrocks}.  To support advanced
consistency conditions (such as snapshot isolation), these engines
have a rich API that offers, in addition to basic put/get operations,
high-level operations such as \emph{atomic multi-key write batches}
and \emph{consistent snapshot iterators}. While several SSD-based \kvx
stores such as RocksDB~\cite{rocksdb}, PebbleDB~\cite{pebbledb}, and
LevelDB~\cite{leveldb} offer this functionality, most existing
NVM-based KV stores do not, making it harder for transactional
databases to benefit from the performance of NVM over SSDs.


One exception to this lack of API support is \prdb, a version of
RocksDB adapted for NVM. Although \prdbX inherits RocksDB's rich API,
it attempts to retrofit its concurrency control and persistence
mechanisms to an SSD-based design. As a result its performance suffers
in comparison to clean slate NVM designs with more restricted
APIs~\cite{listdb,bonsaikv,fluidkv}. In particular, \prdbX relies on a
skiplist-based index to balance the need for fast updates with support
for efficient range queries~\citep{pmemrocksdb}. To
the best of our knowledge, there does not exist a durable skiplist (or
comparable data structure) that both enables RocksDB's rich API and
fully exploits the potential performance gains of NVM.


We propose \sys, an NVM-based storage engine that offers 
high performance
and an API suited for use in modern databases. 
At the heart of \sys is a novel persistent skiplist data structure 
with carefully designed concurrency control and persistence mechanisms
that is optimised for performance (see \cref{sec:concurrency-control}). To
support consistent snapshots while minimising interference between
long running scans and update operations, \sys relies on {\em
  multiversioning}. To facilitate atomic write batches, \sys employs
{\em
  flat-combining}~\cite{DBLP:conf/spaa/HendlerIST10,egorov2025fast,DBLP:conf/usenix/EgorovCDOK24},
which naturally supports batching while also reducing synchronisation
overheads under high contention between update operations.

\sys further enhances performance through a novel four phase execution
framework for update operations. A key design goal is to minimise
synchronisation and persistence delays in the flat-combiner's critical
section.  To achieve this, \sys performs index traversals for update
and read operations, bulk persistence of key-value data, and pointer
updates to non-leaf index nodes in a lock-free manner. Only the
minimal structural modifications required for crash-consistency are
handled within the combiner (e.g. to leaf-index and persistent node
pointers and version numbers). The combiner itself is carefully
optimised through judicious prefetching, compare-and-swap elision, and
a novel persistence optimisation that allows point updates to be
committed using only a single asynchronous fence on the critical path.


Our experimental evaluation using RocksDB's \texttt{db\_bench} benchmarking
tool shows that \sys achieves up to a 73\% increase in throughput over
\prdb, and 75\% over ListDB~\cite{listdb}, a high-performance \kvx store
designed specifically for NVM. Moreover, \sys is designed in a modular fashion
and its durable skiplist can also be readily integrated into existing state-of-the-art \kvx
stores to boost their performance and enrich their functionality. We
demonstrate this via case studies of integrating \sys into \prdb as well as
ListDB. Finally, we rigorously specify and prove the correctness
of \sys, showing that it satisfies durable linearizability, the key safety
guarantee for NVM. 

The remainder of this paper is organised as follows: We first survey the APIs of modern SSD \kvx stores, analyse the limitations of NVM \kvx stores, and introduce our system's memory model (\cref{sec:background}). We then present the design of \sys, including its interface, data layout, and key features of its concurrency control and persistence mechanisms (\cref{sec:design}). This is followed by a detailed description of its update and read operations (\cref{sec:updates}). We then describe how \sys recovers from crashes (\cref{sec:recovery}) and outline our correctness argument (\cref{sec:correctness}). The paper finishes with evaluation results (\cref{sec:performance}), related work (\cref{sec:related-work}) and conclusions (\cref{sec:conclusions}). 

\section{Motivation and Background}
\label{sec:background}
We next discuss how the APIs of modern KV stores have evolved to
support database storage engines~(\cref{sec:key-value-store}), as well
as the shortcomings of state-of-the-art \NVMM KV stores for this
important use case~(\cref{sec:hybrid-nvmm-dram}). Finally, we define
the memory consistency model and the \NVMM persistence primitives that
our algorithms assume in the rest of the
paper~(\cref{sec:system-model}).
 
\subsection{KV Store APIs for Modern Databases}
\label{sec:key-value-store}
\kvx stores (e.g., RocksDB~\cite{rocksdb},
PebbleDB~\cite{pebbledb} and LevelDB~\cite{leveldb}) are
commonly used in standalone deployments but also as
\emph{database storage engines} embedded within 
database management systems. For example, TiDB explicitly delegates
persistence of its storage layer (TiKV) to RocksDB on the grounds that
developing a high-performance standalone storage engine requires
careful and costly optimization~\cite{tidb,tidb_storage}. Similarly,
CockroachDB~\cite{cockroachdb} treats its PebbleDB storage engine as a
black-box API.

To enable high-level database features such as transactions and their
associated concurrency control mechanisms, \kvx storage engines typically
provide an interface that goes beyond basis point operations such as \Op{Put},
\Op{Get} and \Op{Delete}. For example, as shown in Table~\ref{tab:kv-api},
RocksDB, PebbleDB and LevelDB provide a consistent \emph{snapshot iterator}
(\Op{Snapshot}) mechanism that allows higher-level database layers to operate
over a point-in-time view of the data without blocking concurrent
writes~\cite{rocksdb,leveldb,pebbledb}. This can be used as a foundational
building block when implementing transaction isolation levels such as {\em
snapshot isolation} (SI)~\cite{snapshot_isolation}.

In addition, the above \kvx storage engines support \emph{atomic
	multi-key write batch} operations (\Op{WriteBatch}), which allow
multiple updates to be applied atomically. This capability is
also important for database transactions: CockroachDB and TiDB use
\Op{WriteBatch} to coordinate distributed transactions, ensuring that
related updates across multiple keys at a server either all commit or
all abort atomically~\cite{cockroachdb_storage_layer}. 

\begin{table}[t]
\centering
\small
\begin{tabular}{|l|p{6cm}|}
\hline
\textbf{Operation} & \textbf{Functionality} \\
\hline
\textbf{Put} & Inserts a key-value pair and internally assigns a monotonically increasing sequence number for versioning \\
\hline
\textbf{Get} & Point lookups to retrieve current values by key \\
\hline
\textbf{Delete} & Removes keys and their associated values \\
\hline
  \begin{tabular}[t]{@{}l@{}}
    \textbf{Snapshot} \\
    {\bf Iterator}
  \end{tabular}
                   & Forward and backward range scans over key ranges with the ability to provide a consistent point-in-time snapshot view. \\
\hline
	\textbf{Write Batch} & Atomic multi-key batches that apply multiple updates in a single step \\
\hline
\end{tabular}
\caption{Core API operations provided by modern key-value stores such as RocksDB, LevelDB, and PebbleDB. These operations go beyond basic point updates to support versioning, transactions, and consistent iteration.}
\label{tab:kv-api}
\end{table}

\subsection{Hybrid \NVMM-DRAM KV Stores}
\label{sec:hybrid-nvmm-dram}
{\em Non-volatile memory} (\NVMM) is a byte-addressable persistent
storage technology that offers write granularity on the order of
hundreds of bytes (e.g., 256 B rather than 4 KB), providing
significantly lower latency than SSDs or HDDs. In recent years, a
growing body of
research~\cite{fftree,pactree,hikv,flatstore,viper,nvtree,fptree,lbtree,bonsaikv,fluidkv,novelsm,matrixkv,chameleondb,listdb,pmemrocksdb,flatlsm,wb+tree,bztree,dash,slmdb}
has explored how to design \kvx stores that leverage \NVMM to
achieve high performance. Since \NVMM is still slower than DRAM, recent research has shifted towards
\emph{hybrid} designs that keep some indexing information in DRAM for
performance but can also leverage \NVMM for both indexing and data.

A popular hybrid approach is to use a multi-stage/tier indexing
strategy, e.g., based on {\em log-structured merge-trees} (LSM
trees)~\cite{pmemrocksdb, listdb}. The standard LSM architecture 
involves two tiers, a \emph{staging tier} and a
\emph{\capacityTier}. The staging tier consists of an \NVMM log for fast data
persistence and a DRAM resident index (or \emph{memtable}) that may also hold a
copy of the data in the log. The \stagingTier stores recently ingested data
efficiently, while the \capacityTier stores the remaining data set in a more
compact format, for example as sorted string tables (SSTables) on \NVMM.
Background processes periodically migrate data from the staging tier
into the \capacityTier.

However, while many hybrid DRAM-\NVMM KV stores support basic operations
(\PutC/\GetC/\DeleteC), only one (PMemRocksDB~\cite{pmemrocksdb})
supports both \Op{Snapshot} and \Op{WriteBatch} operations (see
\cref{tab:pmem-kv}). This means that (apart
from PMemRocksDB) these systems cannot serve as drop-in replacements
for the storage engines of databases such as CockroachDB and
TiDB~\cite{rocksdb_basic_ops_snapshots,rocksdb_snapshot,cockroachdb}.

\begin{table}[t]
\centering
\small
\begin{tabular}{|l|c|c|c|c|}
\hline
Project & NVM & \makecell{Snapshot \\ Iterator}
				& \makecell{Write \\ Batch} & Speed \\
  \hline
	RocksDB~\cite{rocksdb} & \xmark & \cmark & \cmark & \slow \\
	LevelDB\cite{leveldb} & \xmark & \cmark & \cmark & \slow \\
	PebbleDB~\cite{pebbledb} & \xmark & \cmark & \cmark & \slow \\
	Viper~\cite{viper}  & \cmark & \xmark & \xmark & \fast \\
	BonsaiKV~\cite{bonsaikv}  & \cmark &\xmark/\cmark & \xmark & \fast \\
	FluidKV~\cite{fluidkv}  & \cmark &\xmark/\cmark & \xmark & \fast \\
	ListDB~\cite{listdb}  & \cmark & \xmark & \xmark & \fast \\
	\prdb~\cite{pmemrocksdb}  & \cmark &\cmark & \cmark & \slow \\
    \hline
		\sys & \cmark &\cmark & \cmark & \fast \\
\hline
\end{tabular}
	\caption{Comparison of 
          key-value stores with publicly available implementations. For Snapshot, BonsaiKV and FluidKV are labeled \xmark/\cmark~as they support efficient range queries/scans, but they are not consistent (linearizable)}
\label{tab:pmem-kv}
\end{table}


Designing an NVM \kvx store that correctly
supports both \Op{Snapshot} and \Op{WriteBatch} while maintaining good
performance is fundamentally challenging. Both operations demand
careful coordination as concurrent readers must never observe a partially
executed batch, while snapshot iterators must observe a
consistent point-in-time view. \NVMM adds a further layer of
difficulty: crash safety requires that correctness invariants are
preserved across failures, which recent research has shown is
difficult to get right due to \NVMM's complex low-level
interface~\cite{nvm_bug_detection,nvm_bug_detection_2}. Indeed, in our own work on integrating \sys into
a ListDB-based system we identified two critical correctness bugs in the
original ListDB implementation (see~\cref{sec:isolated_evaluation} and
Appendix~\ref{appendix:listdb_bug}).

\prdbX implements the full RocksDB API by reusing RocksDB's
coordination protocol, including a crash-consistency design that
serializes \NVMM logging and the corresponding DRAM index
updates. This design provides correctness and a familiar API, but it
limits concurrency in the \durableStagingTier: operations are grouped
into a write group, the group leader persists the group's data to the
\NVMM log, and each worker then inserts its data into the
memtable. The write group leader then waits for every worker in the
write group to complete their insertions into the memtable before
allowing any of the threads to return, which limits parallelism.  In a
full \prdbX deployment, this overhead is largely masked: the
\capacityTier is notorious for write stalls under high write
contention~\cite{write_stall}, making it the primary
bottleneck. State-of-the-art \NVMM \kvx stores, however, have largely
eliminated write stalls in the \capacityTier, which shifts attention
to the \durableStagingTier as the new performance-critical component.\looseness=-1

\subsection{System Model}
\label{sec:system-model}

We assume a memory model with a FIFO store buffer (e.g., as in
PTSO~\cite{DBLP:journals/pacmpl/RaadWNV20}) that supports instructions
`{\tt CLWB} $x$', `{\tt SFENCE}' and `{\tt MFENCE}', where
\begin{myitemize}
\item {\tt CLWB} $x$ is a {\em
    cache line write back} that tags the last write on $x$ by the same
  thread for flushing to \NVMM without invalidating the cache line corresponding to $x$, and
\item both {\tt SFENCE} and {\tt MFENCE} block until {\em all} tagged
  writes executed by the same thread have been flushed to \NVMM.
\end{myitemize}
Note that the {\tt SFENCE} instruction may be stored in the
writing thread's store buffer, and only takes effect when the {\tt
  SFENCE} is debuffered (i.e., takes effect in main
memory)~\cite{ptsosyn}. In contrast, {\tt MFENCE} has a stronger
semantics that blocks the executing thread until all cached
instructions (including writes and {\tt CLWB}s) executed by the thread
have taken effect in memory.


{\tt CLWB} and {\tt SFENCE} together are strong enough to support a
common message-passing style synchronisation
pattern~\cite{DBLP:journals/pacmpl/RaadWNV20} demonstrated by the
example in \cref{fig:MP}. Variables $x$, $y$ and $z$ are initialised
to $0$. The left thread updates $x$ to $1$, tags this write using {\tt
  CLWB} $x$ and fences the {\tt CLWB} operation using {\tt SFENCE} before
updating $y$ to $1$. 

The right thread reads $y$, and if it sees the
updated value of $y$ updates $z$ to $1$. This means that the program
satisfies the persistent invariant PInv, i.e., if the value of $z$ in
\NVMM is $1$, then the value of $x$ in \NVMM is also $1$ (but $y$ in
\NVMM may be $0$ or $1$). 

Importantly, visibility of \mbox{$y = 1$} by the
right thread indicates that the $y \gets 1$ has been debuffered in the
left thread. Since the store buffers are FIFO ordered this means that
the previously executed {\tt SFENCE} must have also been debuffered,
which by the memory model semantics means that the write $x \gets 1$ tagged by
{\tt CLWB} $x$ must have been persisted. Note that to get the same guarantee when the code in both columns is executed by the same thread, we must
replace {\tt SFENCE} with {\tt MFENCE}.

\begin{wrapfigure}[11]{r}{0.38\columnwidth}
$\begin{array}[t]{@{}|@{}l@{}|@{}}
  \hline
  ~\text{Init:} x = y = z = 0 \\
  \hline
  \begin{array}[t]{@{}l|@{}l@{}}
    \begin{array}[t]{l}
    x \gets 1; \\ 
    {\tt CLWB}\  x; \\
    {\tt SFENCE}; \\
    y \gets 1
    \end{array}
    & 
    \begin{array}[t]{l}
    a \gets y ; \\
    {\bf if}\ a = 1 \\
    \quad z \gets 1 \\
    \end{array}
  \end{array}\\
  \hline
  \hfill \text{PInv}: z = 1 \Rightarrow x = 1 \hfill{} \\
  \hline
\end{array}$

\caption{Message passing synchronisation pattern}
\label{fig:MP}
\end{wrapfigure}

Finally, for convenience we also define more general operations
\Optt{flush\_range}({\tt ptr,} {\tt size}) and
\Optt{persist\_range}({\tt ptr}, {\tt size}). \Optt{flush\_range}
aligns the address range to cacheline boundaries and flushes all cache
lines that cover the interval [{\tt ptr},\ {\tt ptr} {\tt +} {\tt
  size}] to NVM using the platform's persistence primitives. On our
platform, this is implemented using one or more \texttt{CLWB}
instructions. \Optt{persist\_range} flushes the required range using
\Optt{flush\_range} and then executes an \texttt{SFENCE} instruction
to ensure ordering.

\section{\sysrob Design}
\label{sec:design}
We present \sys, a flexible NVM key-value storage engine that provides a rich API suitable for use in modern databases. To achieve high performance, \sys relies on a hybrid NVM-DRAM skiplist index. The core contributions of \sys centre on the durable skiplist's carefully co-designed concurrency control and persistence mechanisms. 

We next define \sys's interface and correctness guarantees~(\cref{sec:interface}), overview its core data structures~(\cref{sec:core-data-structures}), and outline the key design principles underlying its concurrency control mechanisms~(\cref{sec:concurrency-control}). 


\subsection{Interface Specification}
\label{sec:interface}
\sys provides an API consisting of the following set of operations: 

\begin{description}[leftmargin=0pt,itemsep=2pt,parsep=2pt,font=\normalfont]
	\item [\Op{Get}(\textsc{key}) $\Rightarrow$ \textsc{value}.] Returns the value of the latest version of \texttt{key},  or \texttt{NOT\_FOUND} if no such version exists.

	\item [\Op{Put}(\textsc{key}, \textsc{value})] Stores a new version of the value for the given \texttt{key}. When \Op{Put} returns the new version is durably stored and is visible to subsequent readers. 

	\item [\Op{Delete}(\textsc{key})] Stores a delete marker for \texttt{key}. When \Op{delete} returns the delete entry is durable and visible; readers that observe the delete marker for the key treat it as deleted (i.e., \Op{get} returns \texttt{NOT\_FOUND}). 
	\item [\Op{WriteBatch}(\textsc{ops})] Atomically and durably applies a collection of \Op{put} and \Op{delete} operations. The batch semantics guarantee atomic visibility: once \Op{write\_batch} returns all updates in the batch become visible together (and durably persisted), and readers never observe a partial set of the batch's updates.

	\item [\Op{Snapshot}(\textsc{start}, \textsc{end}) $\Rightarrow$ \Op{Iterator}.] Produces a sequence of key-value pairs that represents a snapshot of the keys in the range [\texttt{start},\texttt{end}]. For convenience this sequence is exposed as an iterator. 


\end{description}

\mypar{Durable linearizability} \sys is designed to implement
\emph{durable linearizability}~\cite{DBLP:conf/wdag/IzraelevitzMS16},
a standard correctness condition for persistent-memory data
structures. Durable linearizability strengthens
linearizability~\cite{DBLP:journals/toplas/HerlihyW90} by additionally
requiring crash safety. Durable linearizability assumes that threads
executing before a crash are not resumed and requires that any history
of operation invocations, responses and system crashes must be
linearizable when the crashes are removed from the history. This means
that any operation that has linearized must also be persisted. Since
operations must linearize before they return, every completed
operation must also be persistent.

\begin{figure}
    \centering
    \includegraphics[width=1\linewidth]{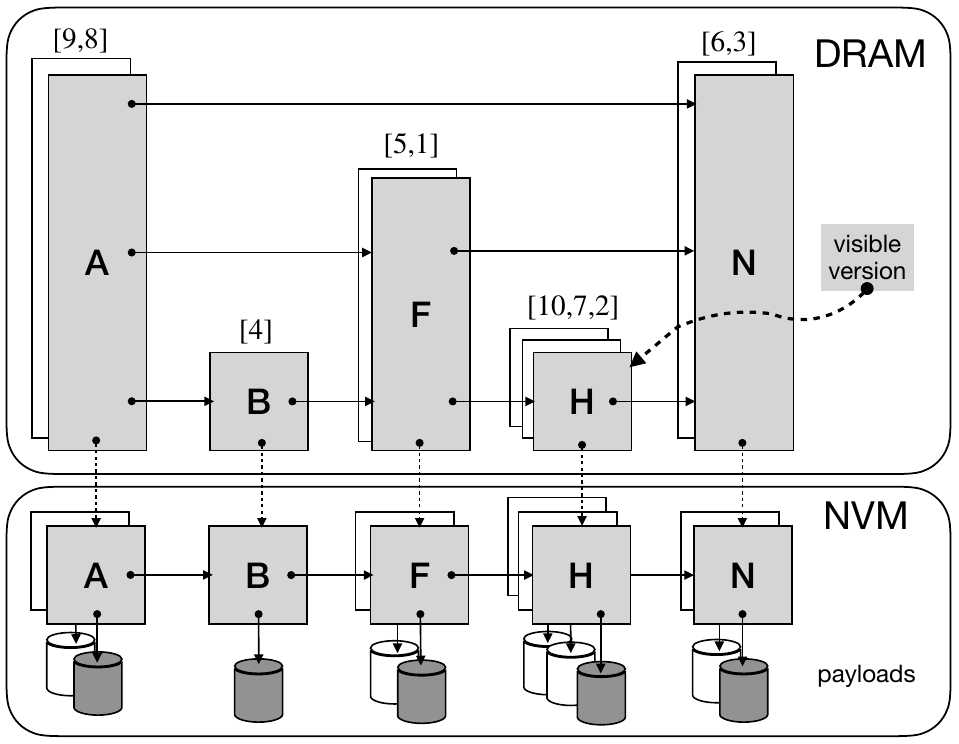}
		\caption{\sys's core data structures are split into volatile and persistent layers (\cref{sec:core-data-structures}). The skiplist nodes are multi-versioned to support \sys's concurrency control. The most recent version across all nodes is stored in \texttt{visible\_version} - in this case 10 for node H (\cref{sec:concurrency-control})}
    \label{fig:core-data-structure}
\end{figure}

\subsection{Core Data Structures\label{sec:core-data-structures}}

At a high-level, \sys consists of a concurrent durable skiplist, as
illustrated in \cref{fig:core-data-structure}. Its state can be divided into two layers, a volatile
layer stored in DRAM containing the skiplist index and a persistent layer
stored in NVM containing the user's key-value data.

\myparr{Volatile Layer:} The volatile (DRAM) layer stores internal index nodes of the skiplist to avoid costly NVM accesses for index traversals and modifications. 
Each volatile index node (\texttt{IndexNode}) contains the user's key, and per-level
successor index pointers. Each index node also contains a pointer to a corresponding persistent node in the persistent layer. 

\myparr{Persistent layer:} The persistent (NVM) layer is responsible for
storing inserted or updated key-value pair nodes durably such that
they can be recovered after a crash. In isolation it constitutes a persistent
linked list. Each persistent node (\texttt{NVMNode}) contains a user
key–value pair and a pointer to its successor persistent node.

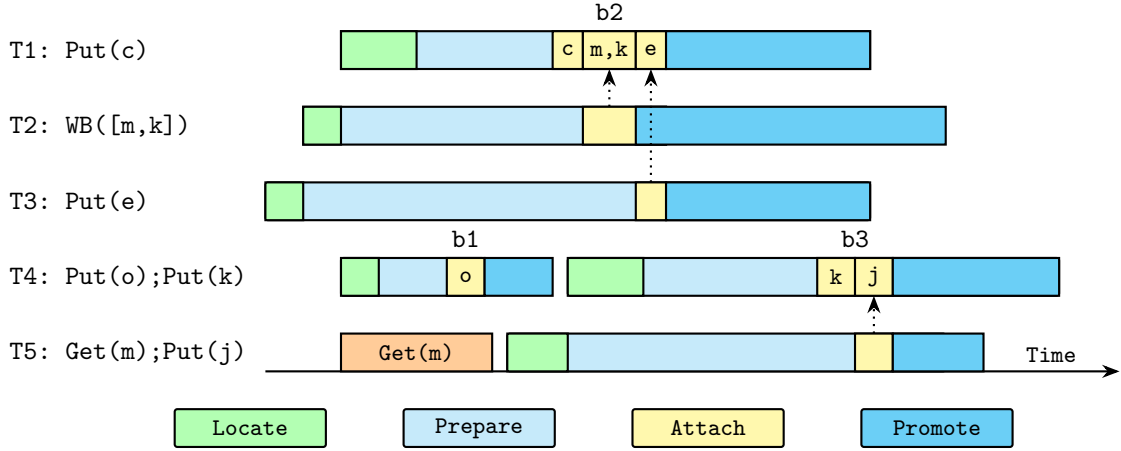
\begin{figure*}[t]
    \centering

 \begin{tikzpicture}[
    font=\ttfamily,
    >=Stealth,
    task/.style={draw,thick},
    op/.style={draw,thick,rounded corners=1pt},
    locate_op/.style={op,fill=green!30,draw=black},
    prepare_op/.style={op,fill=cyan!20,draw=black},
    attach_op/.style={op,fill=yellow!40,draw=black},
    promote_op/.style={op,fill=cyan!50,draw=black},
    get_op/.style={op,fill=orange!40,draw=black},
    locate/.style={task,fill=green!30,draw=black},
    prepare/.style={task,fill=cyan!20,draw=black},
    attach/.style={task,fill=yellow!40,draw=black},
    promote/.style={task,fill=cyan!50,draw=black},
    empty/.style={task,fill=white,draw=black},
    get/.style={task,fill=orange!40,draw=black},
    batchline/.style={black},
    timeline/.style={draw=black,thick,rounded corners=1pt,fill=cyan!20},
    thinsep/.style={gray!40}
]

\def\xmin{0}
\def\xmax{14}

\def\yTOne{4}
\def\yTTwo{3}
\def\yTThree{2}
\def\yTFour{1}
\def\yTFive{0}



	 \draw[->,thick,black] (2.7,-0.25) --node[pos=0.92,black,above] {{\tt Time}} (14,-0.25);


	 \node[anchor=west,font=\large,black] at (-0.8,\yTOne) {{\tt T1: Put(c)}};

	 \node[anchor=west,font=\large,black] at (-0.8,\yTTwo) {{\tt T2: WB([m,k])}};

	 \node[anchor=west,font=\large,black] at (-0.8,\yTThree) {{\tt T3: Put(e)}};

	 \node[anchor=west,font=\large,black] at (-0.8,\yTFour) {{\tt T4: Put(o);Put(k)}};

\node[anchor=west,font=\large,black] at (-0.8,\yTFive) {{\tt T5: Get(m);Put(j)}};

\begin{scope}[xshift=0.7cm]

	 \node[font=\large,black] at (4.65,1.5) {{\tt b1}};
	 \node[font=\large,black] at (6.55,4.5) {{\tt b2}};
	 \node[font=\large,black] at (9.80,1.5) {{\tt b3}};

\draw[prepare] (3,\yTOne-0.25) rectangle (10,\yTOne+0.25);

\draw[locate] (3, \yTOne-0.25) rectangle (4, \yTOne+0.25);

\draw[attach] (5.8,\yTOne-0.25) rectangle (6.2,\yTOne+0.25);
\node at (6,\yTOne) {{\tt c}};

\draw[attach] (6.2,\yTOne-0.25) rectangle (6.9,\yTOne+0.25);
\node at (6.55,\yTOne) {{\tt m,k}};

\draw[attach] (6.9,\yTOne-0.25) rectangle (7.3,\yTOne+0.25);
\node at (7.1,\yTOne) {{\tt e}};

\draw[promote] (7.3,\yTOne-0.25) rectangle (10,\yTOne+0.25);

\draw[prepare] (2.5,\yTTwo-0.25) rectangle (7.3,\yTTwo+0.25);
\draw[attach] (6.2,\yTTwo-0.25) rectangle (7.3,\yTTwo+0.25);
\draw[locate] (2.5,\yTTwo-0.25) rectangle (3.0,\yTTwo+0.25);

\draw[promote] (6.9,\yTTwo-0.25) rectangle (11,\yTTwo+0.25);

\draw[timeline] (2.0,\yTThree-0.25) rectangle (10,\yTThree+0.25);
\draw[locate] (2.0,\yTThree-0.25) rectangle (2.5,\yTThree+0.25);
\draw[attach] (6.9,\yTThree-0.25) rectangle (7.3,\yTThree+0.25);
\draw[promote] (7.3,\yTThree-0.25) rectangle (10,\yTThree+0.25);

\draw[timeline] (3,\yTFour-0.25) rectangle (4.9,\yTFour+0.25);

\draw[timeline] (6,\yTFour-0.25) rectangle (12.5,\yTFour+0.25);
\draw[locate] (3,\yTFour-0.25) rectangle (3.5,\yTFour+0.25);

\draw[attach] (4.4,\yTFour-0.25) rectangle (4.9,\yTFour+0.25);
\node at (4.65,\yTFour) {{\tt o}};

\draw[locate] (6.0,\yTFour-0.25) rectangle (7.0,\yTFour+0.25);

\draw[attach] (9.3,\yTFour-0.25) rectangle (9.8,\yTFour+0.25);
\node at (9.55,\yTFour) {{\tt k}};
\draw[attach] (9.8,\yTFour-0.25) rectangle (10.3,\yTFour+0.25);
\node at (10.05,\yTFour) {{\tt j}};
\draw[promote] (10.3,\yTFour-0.25) rectangle (12.5,\yTFour+0.25);

\draw[promote] (4.9,\yTFour-0.25) rectangle (5.8,\yTFour+0.25);

\draw[get] (3.0,\yTFive-0.25) rectangle (5.0,\yTFive+0.25); 
\node at (4.0,\yTFive) {{\tt Get(m)}};
\draw[timeline] (5.2,\yTFive-0.25) rectangle (11.0,\yTFive+0.25);
\draw[locate] (5.2,\yTFive-0.25) rectangle (6.0,\yTFive+0.25);
\draw[attach] (9.8,\yTFive-0.25) rectangle (10.3,\yTFive+0.25);
\draw[promote] (10.3,\yTFive-0.25) rectangle (11.5,\yTFive+0.25);


\draw[->,thick,black,dotted] (10.05,\yTFive+0.25) -- (10.05,\yTFour-0.25);
\draw[->,thick,black,dotted] (7.1,\yTThree+0.25) -- (7.1,\yTOne-0.25);
\draw[->,thick,black,dotted] (6.55,\yTTwo+0.25) -- (6.55,\yTOne-0.25);

\end{scope}

\begin{scope}[yshift=-1cm]
  \node[locate_op,minimum width=2cm,minimum height=0.5cm] (legL) at (2.5,0) {Locate};
  \node[prepare_op,minimum width=2cm,minimum height=0.5cm,right=1cm of legL] (legP) {Prepare};
  \node[attach_op,minimum width=2cm,minimum height=0.5cm,right=1cm of legP] (legA) {Attach};
  \node[promote_op,minimum width=2cm,minimum height=0.5cm,right=1cm of legA] (legPr) {Promote};
\end{scope}

\end{tikzpicture}

		\caption{The four phases of update operations (\eg Put, WB=WriteBatch) are designed to maximise parallelism across threads. The \phaseLocate, \phasePrepare, and \phasePromote phases of different threads can execute in parallel, while a single combiner thread executes batches of update operations on behalf of itself and other worker threads during the \phaseAttach phase (e.g. batches b1--3 in the figure). Gets and Snapshot reads execute in parallel with all phases of the update operations.}
    \label{fig:timeline_visual}
\end{figure*}

\subsection{Concurrency Control}
\label{sec:concurrency-control}
\sys achieves high performance through careful co-design of its concurrency control and persistence mechanisms. \sys's concurrency control combines multi-versioning with a four-phase execution flow  (\phaseLocate, \phasePrepare, \phaseAttach and \phasePromote) for update operations (\Op{Put}, \Op{Delete} and \Op{WriteBatch}), as illustrated in \cref{fig:timeline_visual} and described in detail in the next section (\cref{sec:updates}). We next overview the core design principles underlying these mechanisms.  

\myparr{Multiversioning.} To ensure durable linearizability while minimizing contention, for example between concurrent write operations and long-running scans, \sys's durable skiplist is multi-versioned, as shown in \cref{fig:core-data-structure}. Each index node includes a version number, which is mirrored in the corresponding persistent node. Version numbers are globally unique, with the exception of atomic \Op{WriteBatch} operations (see \cref{sec:wb}), and incremented for every update. The volatile global variable \texttt{visible\_version} tracks the highest version number used by any update operation. After a crash, \sys reconstructs the pre-crash value of \texttt{visibile\_version} based on the version numbers of persistent nodes (see \cref{sec:recovery}). 

\myparr{Lock-free index traversals.} \sys executes index traversals in a lock-free manner to minimise synchronisation delays when persisting subsequent index modifications. This includes traversals for update operations and read-only operations (\eg~\Op{Get}). For updates an initial lock-free \phaseLocate phase searches optimistically for an appropriate insertion point for the updated node. Similarly, a final lock-free \phasePromote phase lazily updates non-leaf level index node pointers in parallel after index modifications are persisted. 

\myparr{Lock-free bulk persistence.} Given the performance disparity between NVM and DRAM, \sys parallelises bulk persistence operations (e.g. of a new \texttt{NVMNode}'s key and value), which are decoupled from persistence of structural modifications (\eg to pointers and version numbers). Workers perform bulk persistence operations after the \phaseLocate phase as part of a lock-free \phasePrepare phase that allocates, initialises and persists new \texttt{NVMNode}s, but defers linking to them from existing nodes in the persistence layer. 

\myparr{Batch updates using flat-combining.} \sys updates pointers to prepared \texttt{NVMNode}s in the persistent 
layer and the leaf level of the corresponding \texttt{IndexNode} atomically for crash consistency. To reduce
synchronisation overheads under high contention, \sys introduces an \phaseAttach phase that batches these
updates using a flat-combining protocol~\cite{DBLP:conf/spaa/HendlerIST10}.
Workers compete for an exclusive \emph{combiner} lock to modify the base 
layers. The successful combiner thread collects and applies updates on
behalf of itself and the other workers. Due to multi-versioning and the
exclusive lock, no further synchronisation is needed within the combiner
(\eg CAS instructions on pointer updates). To boost concurrency, the
combiner employs an \emph{early release} policy where a waiting worker thread
returns as soon as its prepared node(s) are attached. 

\myparr{Optimising flat-combining for persistence.} \sys introduces several optimisations to reduce persistence delays during
the flat-combining critical section. Since bulk persistence is handled during
the lock-free \phasePrepare phase, the combiner need only persist a minimal amount of critical state (\eg one node
pointer and version number for a \Op{Put} operation). This minimal state is prefetched into cache at the end of the \phasePrepare phase,
which avoids incurring the substantial NVM fetch latency on a cache miss during the \phaseAttach phase. The combiner's persistence operations
are also carefully co-designed with \sys's crash recovery algorithm to require only a single \texttt{SFENCE} for \Op{Put} 
operations (\cref{sec:put}), in contrast to a naive approach requiring separate fences for pointer and version number persistence. 



\newcommand{\inBatch}{in\_batch}

\section{Update and Read Operations}
\label{sec:updates}
We next describe in detail how the design principles outlined in the previous section are realised in the implementation of \sys's update and read operations. We focus first on the description of the \Op{Put}  operation~(\ref{sec:put}), and then briefly outline how it is adjusted to support atomic \Op{WriteBatch}~(\ref{sec:wb}). \Op{Delete} operations insert a tombstone version using the \Op{Put} algorithm, so we omit a separate description here. Finally, we describe briefly \sys's read-only operations (\Op{Get} and \Op{Snapshot}).

\subsection{\Op{Put} Operation}\label{sec:put}

\myparr{Locate:} A \PutC operation begins by executing \phaseLocate,
which conducts a standard lock-free traversal over the index to
collect the predecessor and successor nodes at each level
(\cref{algo:add_outer_shell}, line~\ref{code:search}). This search is
guaranteed to terminate because the skiplist uses sentinel
head and tail nodes.

The \phaseLocate phase executes optimistically, since the \phaseAttach phase may
later discover that the cached neighbours are stale (e.g., a concurrent insertion has
already placed a node at the target location). Therefore, the worker enters a
retry loop that repeatedly executes \phaseLocate~and \phasePrepare followed by
\phaseAttach until the operation
succeeds~(lines~\ref{code:retry_loop_enter}-\ref{code:retry_loop_exit}).

\myparr{Prepare:} On entering the \phasePrepare phase for the first time, (as indicated by a \texttt{NULL} index node parameter),
the worker thread allocates the persistent node and assigns the key and
value (line~\ref{code:pmem_node_create_and_persist}). It subsequently sets the version number to a temporary initial value of \MAXUINT. The final version is assigned later as part our optimised crash recovery algorithm. The worker then flushes the entire memory range occupied by the p\_node to persistent memory (line~\ref{code:p_node_flush_only}), but does not execute \texttt{SFENCE} at this point. The  worker then creates the index node and links it to the p\_node (line~\ref{code:index_node_create_and_link}). These allocation and
persistence steps are performed once per operation and are not
repeated on subsequent retries because they do not depend on the
target location in the skiplist. 

The next part of \Op{prepare} is executed on each invocation. The worker
configures the index node's successor pointers and the persistent node's
successor pointer to point to the cached neighbours collected during
\phaseLocate
(lines~\ref{code:index_node_config_start}-\ref{code:pmem_config_start}). The
worker then calls \persistr, which flushes the cache line containing the
persistent node's successor pointer to persistent memory and executes an
\texttt{SFENCE} to ensure correct ordering
(line~\ref{code:persist_dirty_lines}). To improve performance, it then prefetches the persistent node
into the CPU cache to reduce NVM fetch latency for the combiner thread that will execute
\phaseAttach (line~\ref{code:prefetch_pmem_node}). Finally, it returns the prepared
index node to the caller.

\myparr{Attach (Worker):} Once the nodes are prepared, the worker gets ready for the \phaseAttach phase. It first creates 
a flat-combining advertisement with the information the combiner thread will need~(line~\ref{code:arg_package}).
Advertisements includes the prepared nodes, their neighbours, and a status field set to \texttt{Ready} indicating the request is ready for a combiner thread to process. The worker inserts the advertisement into its own dedicated slot in a global advertisements array, and then enters the \phaseAttach phase (lines~\ref{code:arg_package}-\ref{code:pass_to_flat_combining}). 
Inside \phaseAttach, the worker spins in a short wait loop: it either observes that a combiner has processed its advertisement (line~\ref{code:combiner_finished_check}), or it succeeds in acquiring the flat-combining lock and becomes the combiner. 



\begin{algorithm}[!htbp]
\caption{\label{algo:add_outer_shell}Put operation}
\begin{algorithmic}[1]
\Function{\Op{put}}{key, value}
		\State i\_node $\leftarrow$ \texttt{NULL}
    \While{true} \label{code:retry_loop_enter}
				\State (preds, succs) $\leftarrow$ \textcolor{blue}{\Op{locate}(key)} \label{code:search}
				\State i\_node $\leftarrow$ \textcolor{blue}{\Op{prepare}(key, value, succs, i\_node)}
				\State \gAdvertisements[\LocalThreadId] $\leftarrow$ (i\_node, preds, \texttt{Ready})\label{code:arg_package} 
        \State \textcolor{blue}{\Op{attach}()}\label{code:pass_to_flat_combining} \label{code:flat_combining_invocation}
				\If{\gAdvertisements[\LocalThreadId].status = \texttt{Success}}
            \State \textbf{break}
        \EndIf
    \EndWhile \label{code:retry_loop_exit}
		\State \textcolor{blue}{\Op{promote}(\gAdvertisements[\LocalThreadId])} \label{code:promote}
\EndFunction
\item[]
\Function{\Op{prepare}}{key, value, succs, i\_node}
		\If{i\_node = \texttt{NULL}}
        \State p\_node $\leftarrow$ \Op{create\_pmem\_node}(key, value)\label{code:pmem_node_create_and_persist}
				\State p\_node.version $\leftarrow$ \texttt{MAX\_UINT} \label{code:pmem_seq_init}
								\State \Optt{Flush\_range}(p\_node, p\_node.get\_size() \label{code:p_node_flush_only}
    \State i\_node $\leftarrow$ \Op{create\_index\_node}(key, p\_node)\label{code:index_node_create_and_link}
    \EndIf
    \State i\_node.succs $\leftarrow$ succs \label{code:index_node_config_start}
    \State p\_node.next $\leftarrow$ succs[0].\INbaseNode \label{code:pmem_config_start}
		\State \Optt{persist\_range}(\&p\_node.succ, \texttt{PTR\_SIZE}) \label{code:persist_dirty_lines}
    \State \Op{prefetch}(i\_node.p\_node) \label{code:prefetch_pmem_node}
		\State \Return i\_node
\EndFunction
\item[]
    \Function{attach}{ }
\State my\_adv $\leftarrow$ \gAdvertisements[\LocalThreadId]
\While{true}
\If{my\_adv.status $\in$ \{\texttt{Success}, {\tt Failed}\}} \label{code:combiner_finished_check}
        \State \textbf{return}
    \EndIf

    \If{\Op{try\_acquire\_lock}()} \label{code:fc_lock_acquire}
        \For{adv in \gAdvertisements} \label{code:fc_loop_start}
            \If{adv.status = \texttt{Ready}} \label{code:per_entry_ready_check}
								\State \ProcessPut(adv, false) \label{code:call_process_put}
            \EndIf \label{code:fc_loop_end}
        \EndFor

        \State \Op{lock\_release}()
				\State \Op{MFENCE}() \label{code:mfence}
        \State \textbf{return}
    \EndIf
\EndWhile
\EndFunction

\end{algorithmic}
\end{algorithm}

\myparr{Attach (Combiner):} The combiner scans the advertisements array and processes each advertisement in the \texttt{Ready} state (lines \ref{code:fc_loop_start}--\ref{code:fc_loop_end}).
In rare cases, a thread may acquire the combiner lock before observing that its advertisement has been processed by another combiner. 
In this case the thread will process any \texttt{Ready} advertisements on behalf of other threads, but will skip processing its own to avoid duplicate execution.

Combiner processing of an advertisement is described in Algorithm~\ref{lab:insert_critical}.
The combiner first checks if the cached predecessors and successor are still valid~(line~\ref{code:insert_critical_pred_check}). 
Performing this check under the exclusive combiner lock avoids any need for compare-and-swap instructions when updating the index later in the \phaseAttach phase, as would be required in a regular lock-free skiplist implementation.
If the cached data is invalid the combiner sets the advertisement's status to \texttt{Failed} and returns. 


\begin{algorithm}[!htbp]
	\caption{Combiner processing of a \Op{Put} operation\label{lab:insert_critical}}
\begin{algorithmic}[1]
\Function{\ProcessPut}{adv, \inBatch}
\State new\_index $\leftarrow$ adv.\PAinode
\State new\_pmem $\leftarrow$ new\_index.\INbaseNode
\State succ\_index $\leftarrow$ new\_index.succs[0]
\State pred\_index $\leftarrow$ adv.\PApreds[0]
\If{pred\_index.succs[0] = succ\_index} \label{code:insert_critical_pred_check}
    \State final\_version $\leftarrow$ \visibleSeq + 1\label{code:index_seq_read}
    \State new\_index.version $\leftarrow$ final\_version\label{code:index_seq_assignment}
    \State new\_pmem.version $\leftarrow$ final\_version\label{code:pmem_seq_assignment}
		\State \Optt{Flush\_range}(\&new\_pmem.version, \texttt{UINT\_SIZE})\label{code:pmem_flush_seq_id}
    \State pred\_index.succs[0] $\leftarrow$ new\_index \label{code:index_node_insert}
    \State pred\_pmem $\leftarrow$ pred\_index.\INbaseNode
    \State pred\_pmem.succ $\leftarrow$ new\_pmem \label{code:pmem_node_insert}
		\State \Optt{Persist\_range}(\&pred\_pmem.succ, \texttt{PTR\_SIZE}) \label{code:pmem_flush_pred_next}
    \If{not \inBatch} \label{code:skip_seq_id_increment}
        \State \visibleSeq $\leftarrow$ final\_version\label{code:increment_seq_id}
    \EndIf
		\State adv.status $\leftarrow$ \texttt{Success} 
\Else
		\State adv.status $\leftarrow$ \texttt{Failed} 
\EndIf
\EndFunction
\end{algorithmic}
\end{algorithm}

If the cached data is valid, the combiner assigns a version number greater than the globally visible version number to the index node and NVM node (Algorithm~\ref{lab:insert_critical}, lines~\ref{code:index_seq_assignment}-\ref{code:pmem_seq_assignment}) and flushes the p\_node's version number (line~\ref{code:pmem_flush_seq_id}). 
The combiner then links the index node to its predecessor and the durable node to its predecessor (lines~\ref{code:index_node_insert}-\ref{code:pmem_node_insert}). This makes them visible to threads that are executing \phaseLocate, but not yet to threads that are executing read-only operations.

The successor pointer of the NVM node's predecessor is then persisted using \persistr (line~\ref{code:pmem_flush_pred_next}). At this point the persistent memory node is durably linked to the linked list in persistent memory, and the index node is inserted into the volatile index. Note that in addition to the \texttt{SFENCE} executed within \persistr, a naive solution would perform an additional \texttt{SFENCE} after the earlier \flushr (line~\ref{code:pmem_flush_seq_id}) to ensure the correct node version number is persisted before the node is durably linked. However, our recovery algorithm allows us to ensure correctness using only a single \texttt{SFENCE} (see \cref{sec:recovery}), substantially improving the combiner's performance.

The globally visible version number is now safely incremented and the new node becomes visible to read-only operations~(line~\ref{code:increment_seq_id}).
Finally, the combiner sets the advertisement's status to \texttt{Success} and returns.  After releasing its lock, the combiner executes an \texttt{MFENCE} to ensure its own operation has been flushed to NVM (Algorithm \ref{algo:add_outer_shell}, line~\ref{code:mfence}).

\myparr{Promote:} While the combiner continues processing any remaining ready entries in the advertisements array, a waiting worker that observes its advertisement's status is \texttt{Success} immediately proceeds from the \phaseAttach phase to the \phasePromote phase~(\cref{algo:add_outer_shell}, line~\ref{code:promote}) (or conversely returns to the \phasePrepare phase if it is \texttt{Failed}).

During \phasePromote (not shown), the worker updates higher levels of the volatile index concurrently with other threads using a lock-free skiplist algorithm. The promotion inserts the new index node into levels in ascending order, which preserves the invariant that if the node exists in level \textit{i} then it also exists in every level below \textit{i}.

\subsection{\Op{WriteBatch} Operation}
\label{sec:wb} The \Op{WriteBatch} operation atomically and durably applies a
collection of \Op{put} and \Op{delete} operations~ using the same four-phase
execution framework as \Op{Put}. We next outline the key differences. 

\myparr{Locate and Prepare:} 
In contrast to \Op{Put}, \Op{WriteBatch} takes a vector of key-value pairs as an argument~(Algorithm~\ref{lab:insert_batch_outer_shell},~line~\ref{code:wb_params}). Worker threads therefore execute \phaseLocate and \phasePrepare for each individual pair, collecting the resulting prepared nodes into the set \texttt{batchAdvs} (lines~\ref{code:batch_advs_start}-\ref{code:batch_advs_end}). Unlike \Op{Put}, the worker thread never retries \phaseLocate and \phasePrepare for operations in a \Op{WriteBatch}. This avoids having to retry an entire batch for a single failed operation. Instead, the combiner retries failed operations (at most once) during the \phaseAttach phase. 

\myparr{Attach (Worker):} The worker thread nexts creates a flat-combining advertisement prior to entering the \phaseAttach phase (line~\ref{code:batch_adv}). A minor variation here from \Op{Put} is that a \Op{WriteBatch} advertisement contains a collection \texttt{batchAdvs} of prepared nodes, one for each operation, as well as a status field for the whole batch.

\myparr{Attach (Combiner):} The overall \phaseAttach phase flat-combining flow is the same as for \Op{Put}, except that for \Op{WriteBatch} advertisements the combiner executes \Op{Process\_WriteBatch} instead of \Op{Process\_Put}~(Algorithm~\ref{algo:add_outer_shell}, line~\ref{code:call_process_put}).  

Within \Op{Process\_WriteBatch}, the key challenge for the combiner is to
ensure crash-consistency for the batch as a whole (Algorithm~\ref{algo:process_wb}). For this it relies on two
persistent variables: \commitSeqIdt to record the version of the most recent
completed operation before the start of the batch, and a \useCommitEntryt flag
to record whether the combiner was processing a batch when
a crash occured. Before processing batch operations, the combiner first updates
\commitSeqIdt with the value of \visibleSeqt and persists
it~(lines~\ref{code:start_batch}-\ref{code:flush_commit_seq_id}). It then sets
and persists the \useCommitEntryt flag to switch to batch
mode~(lines~\ref{code:set_use_commit_entry}-\ref{code:flush_use_commit_entry}).
If upon recovery the system observes that the \useCommitEntryt flag is set, it
will roll back any nodes with a version higher than \commitSeqIdt. Conversely,
if the \useCommitEntryt flag is not set, it will recover all commited nodes.  

For each prepared node in the batch, the combiner invokes \Op{Process\_Put} to insert the
node into the data structure (lines~\ref{code:writrbatch_retry_start}-\ref{code:wb_end_for}).
All nodes in the batch are assigned the same version number in \Op{WriteBatch} because the
\visibleSeqt is not incremented for batch operations
(Algorithm~\ref{lab:insert_critical}, line~\ref{code:skip_seq_id_increment}).
Similarly to the \Op{Put} operation, a retry may be required if the cached
neighbours of the prepared node have become stale. However, in the case of
\Op{WriteBatch} the retry is executed within \phaseAttach to
ensure that all operations in the batch have completed. There can be at most one retry per entry, since \Op{Process\_Put} executes
under the exclusive combiner lock. 

Once all operations have executed, the combiner switches back to non-batch mode
(line~\ref{code:commit_batch}). At this point the batch is considered committed
and the combiner increments \linebreak[1]\visibleSeqt to make the batch's nodes visible
atomically to readers~(line~\ref{code:batch_seq_id_assignment}). 

\myparr{Promote:} Finally, after the worker observes the batch advertisement's status is
\texttt{Success}, it executes the \phasePromote phase for each operation in the
batch~(Algorithm~\ref{lab:insert_batch_outer_shell},
lines~\ref{code:wb_promote_for}-\ref{code:wb_promote_endfor}). 

\begin{algorithm}[htbp]
    \caption{\WriteBatch Operation\label{lab:insert_batch_outer_shell}}
    \begin{algorithmic}[1]
\Function{\WriteBatch}{kv\_pairs}\label{code:wb_params}
			\For{(key, value) in kv\_pairs} \label{code:batch_advs_start}
				\State (preds, succs) $\leftarrow$ \textcolor{blue}{\Op{locate}(key)} \label{code:batch_locate}
				\State i\_node $\leftarrow$ \textcolor{blue}{\Op{prepare}(key, value, succs, \texttt{NULL})} \label{code:batch_prepare}
			\State batchAdvs.append((i\_node, preds, \texttt{Ready}))\label{code:batch_package}
				\EndFor \label{code:batch_advs_end}
				\State \gAdvertisements[\LocalThreadId] $\leftarrow$ (batchAdvs, \texttt{Ready}) \label{code:batch_adv}
				\State \textcolor{blue}{\Op{attach}()}\label{code:write_batch_attach}
				\For{putAdv in batchAdvs} \label{code:wb_promote_for}
					 \State \textcolor{blue}{\Op{promote}(putAdv)}
				\EndFor \label{code:wb_promote_endfor}
\EndFunction
\end{algorithmic}
\end{algorithm}

\begin{algorithm}[htbp]
\caption{Combiner \Op{WriteBatch} processing \label{algo:process_wb}}
\begin{algorithmic}[1]
	\Function{\Op{Process\_WriteBatch}}{wbAdv}
\State \commitSeqId $\leftarrow$ \visibleSeq \label{code:start_batch}
	\State \Optt{persist\_range}(\&\commitSeqId, \texttt{UINT\_SIZE}) \label{code:flush_commit_seq_id} 
\State \useCommitEntry $\leftarrow$ true \label{code:set_use_commit_entry}
	\State \Optt{persist\_range}(\&\useCommitEntry, \texttt{BOOL\_SIZE}) \label{code:flush_use_commit_entry}
\For{putAdv in wbAdv.batchAdvs}
\While{true} \label{code:writrbatch_retry_start}
\State \Op{Process\_Put}(putAdv, true)
\If{putAdv.status = \texttt{Success}} \label{code:invoke_insert_critical}
\State \textbf{break};
\EndIf
	\State key $\leftarrow$ putAdv.i\_node.key
	\State (putAdv.preds, succs) $\leftarrow$ \textcolor{blue}{\Op{Locate}(key)}
	\State value $\leftarrow$ putAdv.i\_node.p\_node.value
	\State putAdv.i\_node $\leftarrow$ 
			 \State \quad \textcolor{blue}{\Op{prepare}(key, value, \PAsuccs, putAdv.i\_node)}

\EndWhile \label{code:writrbatch_retry_end}
	\EndFor \label{code:wb_end_for}

\State \useCommitEntry $\leftarrow$ false \label{code:commit_batch}
	\State \Optt{persist\_range}(\&\useCommitEntry, \texttt{BOOL\_SIZE}) \label{code:flush_use_commit_entry_end}
\State \visibleSeq $\leftarrow$ \commitSeqId + 1 \label{code:batch_seq_id_assignment}
\State wbAdv.status $\leftarrow$ \texttt{Success}
\EndFunction
\end{algorithmic}
\end{algorithm}

\subsection{\Op{Get} and \Op{Snapshot} Operations\label{sec:detailed_algo_finish}}
\Op{Get} operations follow a standard lock-free skiplist traversal algorithm, with one addition: they ignore nodes that are not yet visible to the calling thread. This visibility control is achieved by comparing each encountered node's version against the caller's active view based on the value of \texttt{visible\_version}.

The \Op{Get} operation begins by caching the current visible version on (Algorithm~\ref{algo:get_operation}, line~\ref{code:get_op_sid_cache}). It then traverses from the head node down the levels of the skiplist to locate the target position of \lookupKey (lines~\ref{code:get_op_traverse}-\ref{code:end_get_op_traverse}). During this traversal, the inner loop advances along the current level as long as the version-aware key comparator \Op{Cmp} determines that the node's key is strictly less than the target key given the cached visible version (\texttt{vv}) (lines~\ref{code:get_op_compare}-\ref{code:end_get_op_compare}).

When the lock-free traversal finishes, the algorithm evaluates whether the target key was found. First, if the successor node's key matches the target key (line~\ref{code:check_match}), it returns the associated value. Otherwise it returns \texttt{NULL}.
\begin{algorithm}[!htbp]
\caption{\label{algo:get_operation}Get Operation}
\begin{algorithmic}[1]

\Function{\Op{get}}{\lookupKey}
    \State vv $\leftarrow$ visible\_version\label{code:get_op_sid_cache}
    \State pred $\leftarrow$ head
		\For{i $\leftarrow$ \texttt{MAX\_HEIGHT} - 1 \textbf{downto} 0} \label{code:get_op_traverse}
        \State curr $\leftarrow$ pred.succs[i]
        \While{curr $\neq$ tail \textbf{and} \label{code:get_op_compare}
				\State \qquad \Op{Cmp}(curr.key, \lookupKey, vv) < 0} 
            \State pred $\leftarrow$ curr
            \State curr $\leftarrow$ pred.succs[i]
				\EndWhile \label{code:end_get_op_compare}
    \EndFor \label{code:end_get_op_traverse}
    \If{curr.succs[0].key = \lookupKey} \label{code:check_match}
        \State \Return curr.p\_node.value
    \Else
				\State \Return \texttt{NULL} 
    \EndIf
\EndFunction

\end{algorithmic}
\end{algorithm}

The \Op{Snapshot} operation proceeds as follows. The thread first caches the \texttt{visible\_version} and then traverses the skip list until it finds the first node whose key is greater than or equal to the snapshot’s start key and whose version is less than or equal to the cached \texttt{visible\_version}, or until it reaches the tail, in which case it returns \texttt{NULL}.
If the specified range is present in the skip list, the operation returns an iterator object to the caller, which allows traversal of the skip list from that point onward. Each call to the iterator’s \texttt{next()} method returns the node currently pointed to by the iterator and advances the iterator to the next valid node.
A node is considered valid if its key is greater than the current node’s key and its version is less than or equal to the cached \texttt{visible\_version}. If no such node exists, the iterator will set the return value for the next invocation to \texttt{NULL}.

\section{Recovery Procedure}
\label{sec:recovery}

In the event of a crash, the DRAM index and visible version number are lost, while a durable sorted linked list remains on persistent memory. Two scenarios are possible during recovery:

\begin{myenumerate}
	\item \textbf{No batch in progress.} If the \texttt{batch\_mode} flag is set
		to \textit{false}, recovery simply collects all nodes that were present in
		the list, except potentially for a single \texttt{NVMNode} whose version is
		equal to \texttt{MAX\_UINT}, indicating that the node has not yet been
		fully committed and was not observable before the crash event.

    \item \textbf{Batch in progress.} Otherwise, recovery checks the version number of the last committed operation and excludes nodes with a higher version from the recovery process.
\end{myenumerate}

The list can then be traversed and used to re-initialise the visible version number and to rebuild the in-memory index according to the surrounding system implementation. For example, this may involve reconstructing a memtable in \prdbX or inserting nodes into the L0 skiplist in the ListDB implementation.

From a performance perspective, the key difference between recovery in \sys and other implementations is that the persistent list is already sorted. This property enables faster reconstruction of the skiplist index during recovery, as we demonstrate in our experimental evaluation (\cref{sec:recovery_time_evaluaiton}).

\section{Correctness of \sysrob}
\label{sec:correctness}
In this section, we give a brief overview of the proof that \sys satisfies 
durable linearizability. 
A complete proof can be found in Appendix~\ref{sec:dl-proof}.

We represent key/value pairs stored in \sys as a collection of
\mbox{read/write} variables, and model the update operations (\Op{Put}, \Op{WriteBatch}, 
and \Op{Delete}) as \emph{atomic writes} of values to one or more keys, and
the read operations (\Op{Get} and \Op{Snapshot}) as \emph{atomic reads}
from one or more keys. 

Our proof strategy follows
the \emph{dependency graph} approach of~\cite{adya99-weak-consis}, 
which asserts that a history $\sigma$ of operations is linearizable 
provided a directed graph induced by the union of four relations -- $\WR$,
$\WW$, $\RW$, and $\RT$ -- is acyclic (Theorem~\ref{thm:linearizability}). 
Informally, $\WR$ (reads-from) identifies pairs $r$ and $w$ of read and write operations such that 
$r$ returns the values written by $w$, $\WW$ (write-write)
establishes a total order over all write operations in $\sigma$, 
and $\RT$ represents the real-time invocation order between pairs of
operations in $\sigma$. The relation $\RW$ (from-read) is derived
from $\WR$ and $\WW$ (Definition~\ref{dep-graph-def}). 

To accommodate the operations that may remain incomplete due to crashes,
we introduce a \emph{visibility} predicate~\cite{artyom} such that
a read operation $r$ is visible iff the thread that invoked $r$
does not crash before $r$ returns; and 
a write operation $w$
is visible iff the thread that invoked 
$w$ does not crash before the assignment in line~\ref{code:pmem_node_insert} 
(if $w =$ \Op{Put}) or line~\ref{code:commit_batch}
(if $w =$ \Op{WriteBatch}) reaches persistence.
We then restrict the four relations
above to hold only for visible operations (Definition~\ref{dep-graph-def}),
and revise the linearizability criterion as in~\cite{artyom} to assert
durable linearizability (Theorem~\ref{thm:linearizability}). 

In the remainder of the proof, we show how the witnesses for the relations
$\WR$ and $\WW$ can be constructed for a given history of \sys 
(Definition~\ref{def:tau-tag}),
and prove that the resulting dependency graph is acyclic 
(\cref{thm:dep-graph-acyc}). 
The durable linearizability of \sys then follows from 
Theorem~\ref{thm:linearizability}.

\section{Evaluation}\label{sec:performance}
In this section, we evaluate \sys's performance to show the benefits of its
concurrency control and persistence mechanisms. Our evaluation goals are
threefold. First, we investigate the throughput and latency of \sys's
durable skiplist index in isolation to understand its behaviour in a standalone
deployment (Section~\ref{sec:isolated_evaluation}). We then integrate \sys
with the capacity tiers of existing key-value stores and evaluate its ability
to boost their end-to-end performance (Section~\ref{sec:end_to_end}).  Finally,
we evaluate the time taken by \sys for crash recovery
(Section~\ref{sec:recovery_time_evaluaiton}).

\subsection{Experimental Setup\label{sec:experimental_setup}}
We conduct all experiments on a dual-socket server equipped with two Intel Xeon Gold 6326 processors (16 cores per socket, 2 hardware threads per core, 2.90 GHz base/3.50 GHz boost) and 128 GB of DRAM. Persistent memory is provided by a 496 GiB Intel Optane PMEM module mounted in fsdax mode on the first NUMA node. All benchmarks are pinned to a single NUMA node (\texttt{numactl --cpunodebind=1 --membind=1}) to eliminate cross-socket effects. The system runs Ubuntu 22.04 with Linux kernel 6.8.0, and all implementations are compiled with GCC 11.4 at (\texttt{-O3 -mavx2 -msse4.2 -mclflushopt -mclwb}). 
Our evaluation uses the standard benchmarks provided by the RocksDB \texttt{db\_bench} tool. Unless otherwise specified, all experiments are conducted using the \texttt{fillrandom} workload.

\begin{figure*}[ht]
  \centering
  \begin{subfigure}[t]{0.31\textwidth}
    \centering
    \begin{tikzpicture}
    \begin{axis}[
      bench style,
      xlabel={Thread count},
      ylabel={Throughput (Mops/s)},
      ymin=0.0,
      legend style={
      at={(1,0)},
      anchor=south east,
    }
    ]
    \addplot[color=veliriscolor, dashed, mark=square*]
      coordinates{(2,1.259) (4,1.768) (6,2.113) (8,2.435) (10,2.715) (12,2.972) (14,3.180) (16,3.302) (18,3.427) (20,3.546) (22,3.601) (24,3.645) (26,3.678) (28,3.722) (30,3.748) (32,3.759)};
    \addlegendentry{\sys}

    \addplot[color=listdbcolor, solid, mark=*]
      coordinates{(2,1.530) (4,2.096) (6,2.302) (8,2.359) (10,2.305) (12,2.244) (14,2.197) (16,2.183) (18,2.186) (20,2.248) (22,2.251) (24,2.249) (26,2.249) (28,2.241) (30,2.226) (32,2.210)};
    \addlegendentry{ListDB}

    \addplot[color=groupcommitcolor, dash dot, mark=triangle*]
      coordinates{(2,1.094) (4,1.547) (6,1.944) (8,2.280) (10,2.551) (12,2.730) (14,2.859) (16,2.907) (18,3.001) (20,3.020) (22,3.090) (24,3.091) (26,3.077) (28,3.063) (30,3.018) (32,2.991)};
    \addlegendentry{PmemRocksDB}
    \end{axis}
    \end{tikzpicture}
    \caption{num\_ops\,=\,300K\label{fig:bench_throughput_1}}
  \end{subfigure}
  \hfill
  \begin{subfigure}[t]{0.31\textwidth}
    \centering
    \begin{tikzpicture}
    \begin{axis}[
      bench style,
      xlabel={Thread count},
      ylabel={Throughput (Mops/s)},
      ymin=0.0,
      legend entries={},
    ]

    \addplot[color=veliriscolor, dashed, mark=square*]
      coordinates{(2,1.021) (4,1.561) (6,1.956) (8,2.274) (10,2.521) (12,2.774) (14,3.023) (16,3.201) (18,3.349) (20,3.471) (22,3.553) (24,3.612) (26,3.655) (28,3.675) (30,3.731) (32,3.750)};

    \addplot[color=listdbcolor, solid, mark=*]
      coordinates{(2,1.274) (4,1.901) (6,2.220) (8,2.301) (10,2.289) (12,2.252) (14,2.194) (16,2.177) (18,2.152) (20,2.244) (22,2.262) (24,2.251) (26,2.249) (28,2.243) (30,2.233) (32,2.217)};

    \addplot[color=groupcommitcolor, dash dot, mark=triangle*]
      coordinates{(2,0.844) (4,1.335) (6,1.731) (8,2.072) (10,2.341) (12,2.540) (14,2.702) (16,2.819) (18,2.883) (20,2.956) (22,3.005) (24,3.055) (26,3.063) (28,3.042) (30,3.062) (32,3.025)};
    \end{axis}
    \end{tikzpicture}
    \caption{num\_ops\,=\,1M\label{fig:bench_throughput_2}}
  \end{subfigure}
  \hfill
  \begin{subfigure}[t]{0.31\textwidth}
    \centering
    \begin{tikzpicture}
    \begin{axis}[
      bench style,
      xlabel={Thread count},
      ylabel={Throughput (Mops/s)},
      ymin=0.0,
    ]

    \addplot[color=veliriscolor, dashed, mark=square*]
      coordinates{(2,0.877) (4,1.431) (6,1.824) (8,2.148) (10,2.403) (12,2.632) (14,2.889) (16,3.087) (18,3.253) (20,3.401) (22,3.486) (24,3.579) (26,3.622) (28,3.661) (30,3.694) (32,3.724)};

    \addplot[color=listdbcolor, solid, mark=*]
      coordinates{(2,1.100) (4,1.729) (6,2.088) (8,2.276) (10,2.303) (12,2.261) (14,2.213) (16,2.186) (18,2.135) (20,2.224) (22,2.244) (24,2.241) (26,2.245) (28,2.239) (30,2.233) (32,2.218)};

    \addplot[color=groupcommitcolor, dash dot, mark=triangle*]
      coordinates{(2,0.697) (4,1.164) (6,1.539) (8,1.871) (10,2.143) (12,2.362) (14,2.526) (16,2.660) (18,2.725) (20,2.819) (22,2.907) (24,2.900) (26,2.970) (28,3.001) (30,3.032) (32,3.038)};
  
    \end{axis}
    \end{tikzpicture}
    \caption{num\_ops\,=\,2M\label{fig:bench_throughput_3}}
  \end{subfigure}
	\caption{Throughput \vs \# Cores for \texttt{db\_bench} fillrandom benchmark (value size 64B).}
\end{figure*}
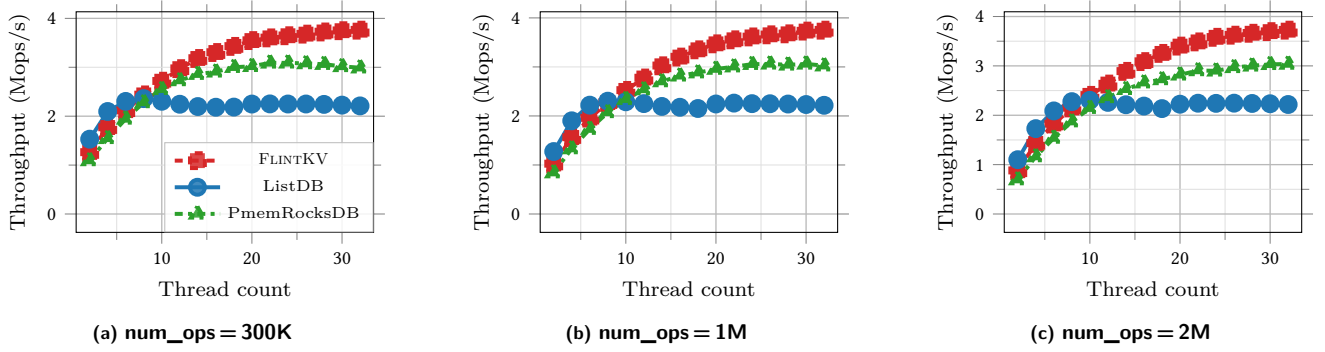

\begin{figure*}[ht]
  \centering
  \begin{subfigure}[t]{0.31\textwidth}
    \centering
    \begin{tikzpicture}
    \begin{axis}[
      bench style,
      xlabel={Thread count},
      ylabel={Latency (ns)},
      ymin=0.0,
      ymax=20000,
    ]

    \addplot[color=veliriscolor, dashed, mark=square*]
      coordinates{(2,1898.54) (4,2427.48) (6,2859.19) (8,3177.10) (10,3438.32) (12,3623.53) (14,3759.46) (16,3954.39) (18,4322.14) (20,4712.48) (22,5163.78) (24,5515.62) (26,6000.38) (28,6452.74) (30,6774.96) (32,7393.36)};
    \addlegendentry{\sys}

    \addplot[color=listdbcolor, solid, mark=*]
      coordinates{(2,1563.08) (4,1935.99) (6,2306.18) (8,2848.04) (10,3580.43) (12,4289.22) (14,4912.11) (16,5520.61) (18,5220.13) (20,5054.94) (22,5543.85) (24,6141.09) (26,6832.20) (28,7615.11) (30,8384.73) (32,9390.92)};
    \addlegendentry{ListDB}
    
    \addplot[color=groupcommitcolor, dash dot, mark=triangle*]
      coordinates{(2,2255.00) (4,2846.49) (6,3248.56) (8,3596.84) (10,3922.66) (12,4272.21) (14,4692.08) (16,5146.29) (18,5640.41) (20,6115.89) (22,6670.98) (24,7112.66) (26,7669.50) (28,8095.57) (30,8624.13) (32,9102.79)};
    \addlegendentry{PmemRocksDB}
    \end{axis}
    \end{tikzpicture}
    \caption{Median Latency (P50)\label{fig:bench2_p50}}
  \end{subfigure}
  \hfill
  \begin{subfigure}[t]{0.31\textwidth}
    \centering
    \begin{tikzpicture}
    \begin{axis}[
      bench style,
      xlabel={Thread count},
      ylabel={Latency (ns)},
      ymin=0.0,
      ymax=20000,
    ]

    \addplot[color=veliriscolor, dashed, mark=square*]
      coordinates{(2,3988.70) (4,4883.40) (6,5784.78) (8,6497.57) (10,7180.06) (12,7927.95) (14,8753.10) (16,9404.40) (18,9924.10) (20,10703.59) (22,11560.61) (24,11897.00) (26,12758.30) (28,13578.67) (30,13896.52) (32,15009.63)};

    \addplot[color=listdbcolor, solid, mark=*]
      coordinates{(2,4078.82) (4,5381.86) (6,5999.51) (8,7832.86) (10,11622.33) (12,16458.97) (14,21901.92) (16,26698.63) (18,36188.96) (20,42835.80) (22,47621.41) (24,52008.40) (26,56141.19) (28,59637.74) (30,65988.41) (32,73004.35)};

    \addplot[color=groupcommitcolor, dash dot, mark=triangle*]
      coordinates{(2,4727.73) (4,5735.72) (6,6312.92) (8,6924.62) (10,7709.88) (12,8326.42) (14,8969.84) (16,9800.55) (18,11300.06) (20,12108.91) (22,13569.92) (24,14580.59) (26,15839.70) (28,17429.28) (30,18475.08) (32,19799.73)};
    
    \end{axis}
    \end{tikzpicture}
    \caption{99th-pct.\ Latency (P99)\label{fig:bench_2_p99}}
  \end{subfigure}
  \hfill
  \begin{subfigure}[t]{0.31\textwidth}
    \centering
    \begin{tikzpicture}
    \begin{axis}[
      bench style,
      xlabel={Thread count},
      ylabel={Tail amplification (P99\,/\,P50)},
      ymin=0.0,
    ]
    \addplot[color=veliriscolor, dashed, mark=square*]
      coordinates{(2,2.10) (4,2.01) (6,2.02) (8,2.05) (10,2.09) (12,2.19) (14,2.33) (16,2.38) (18,2.30) (20,2.27) (22,2.24) (24,2.16) (26,2.13) (28,2.10) (30,2.05) (32,2.03)};
    
    \addplot[color=listdbcolor, solid, mark=*]
      coordinates{(2,2.61) (4,2.78) (6,2.60) (8,2.75) (10,3.25) (12,3.84) (14,4.46) (16,4.84) (18,6.93) (20,8.47) (22,8.59) (24,8.47) (26,8.22) (28,7.83) (30,7.87) (32,7.77)};

    \addplot[color=groupcommitcolor, dash dot, mark=triangle*]
      coordinates{(2,2.10) (4,2.02) (6,1.94) (8,1.93) (10,1.97) (12,1.95) (14,1.91) (16,1.90) (18,2.00) (20,1.98) (22,2.03) (24,2.05) (26,2.07) (28,2.15) (30,2.14) (32,2.18)};
    
    \end{axis}
    \end{tikzpicture}
    \caption{Tail amplification (P99/P50)\label{fig:bench2_tail}}
  \end{subfigure}
	\caption{Latency \vs \# Cores for \texttt{db\_bench} fillrandom benchmark (300k operations, value size 64B).}
\end{figure*}
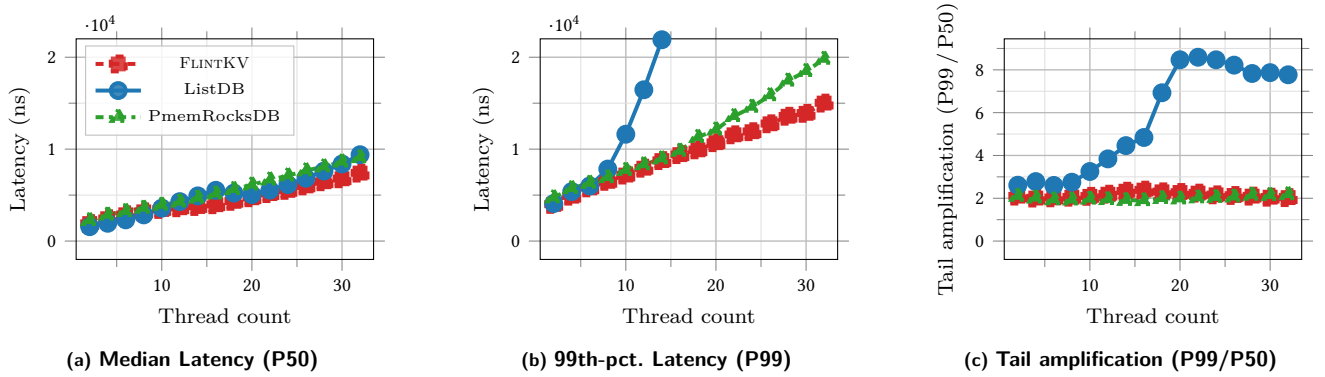

\medskip
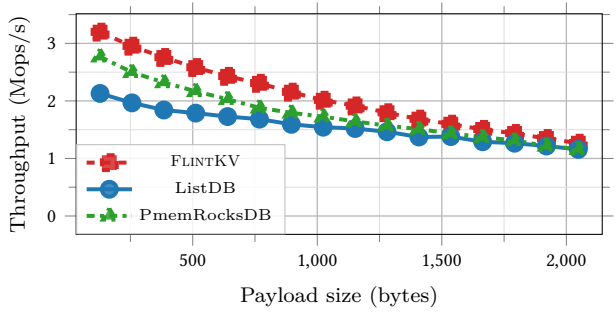
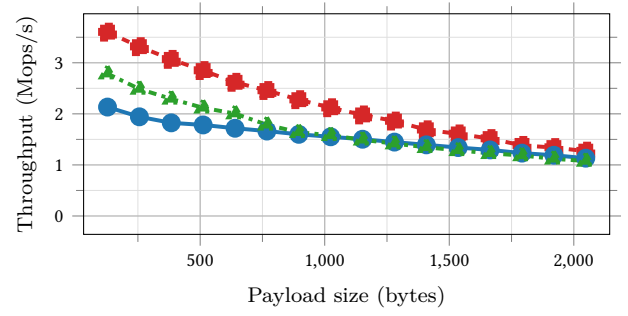
\begin{figure*}[ht]
  \centering
  \begin{subfigure}[t]{0.48\textwidth}
    \centering
    \begin{tikzpicture}
    \begin{axis}[
      bench style,
      xlabel={Payload size (bytes)},
      ylabel={Throughput (Mops/s)},
      ymin=0.0,
      legend style={
      at={(0,0)},
      anchor=south west,
    }
    ]
    \addplot[color=veliriscolor, dashed, mark=square*]
      coordinates{(128,3.195) (256,2.951) (384,2.755) (512,2.580) (640,2.427) (768,2.306) (896,2.152) (1024,2.012) (1152,1.913) (1280,1.795) (1408,1.684) (1536,1.596) (1664,1.505) (1792,1.435) (1920,1.347) (2048,1.265)};
    \addlegendentry{\sys}
    \addplot[color=listdbcolor, solid, mark=*]
      coordinates{(128,2.130) (256,1.966) (384,1.842) (512,1.788) (640,1.728) (768,1.683) (896,1.593) (1024,1.543) (1152,1.520) (1280,1.464) (1408,1.373) (1536,1.380) (1664,1.290) (1792,1.265) (1920,1.216) (2048,1.157)};
    \addlegendentry{ListDB}
    
    \addplot[color=groupcommitcolor, dash dot, mark=triangle*]
      coordinates{(128,2.768) (256,2.498) (384,2.320) (512,2.167) (640,2.029) (768,1.880) (896,1.793) (1024,1.726) (1152,1.643) (1280,1.570) (1408,1.509) (1536,1.436) (1664,1.371) (1792,1.312) (1920,1.218) (2048,1.161)};
    \addlegendentry{PmemRocksDB}
    \end{axis}
    \end{tikzpicture}
    \caption{Threads\,=\,16\label{fig:bench3_1}}
  \end{subfigure}
  \hfill
  \begin{subfigure}[t]{0.48\textwidth}
    \centering
    \begin{tikzpicture}
    \begin{axis}[
      bench style,
      xlabel={Payload size (bytes)},
      ylabel={Throughput (Mops/s)},
      ymin=0.0,
    ]

    \addplot[color=veliriscolor, dashed, mark=square*]
      coordinates{(128,3.600) (256,3.316) (384,3.067) (512,2.848) (640,2.618) (768,2.457) (896,2.278) (1024,2.120) (1152,1.979) (1280,1.859) (1408,1.688) (1536,1.607) (1664,1.516) (1792,1.382) (1920,1.340) (2048,1.270)};

    \addplot[color=listdbcolor, solid, mark=*]
      coordinates{(128,2.133) (256,1.942) (384,1.824) (512,1.783) (640,1.718) (768,1.663) (896,1.599) (1024,1.551) (1152,1.503) (1280,1.447) (1408,1.393) (1536,1.340) (1664,1.294) (1792,1.231) (1920,1.189) (2048,1.132)};

    \addplot[color=groupcommitcolor, dash dot, mark=triangle*]
      coordinates{(128,2.787) (256,2.487) (384,2.293) (512,2.118) (640,2.005) (768,1.786) (896,1.639) (1024,1.581) (1152,1.492) (1280,1.417) (1408,1.345) (1536,1.285) (1664,1.227) (1792,1.185) (1920,1.119) (2048,1.077)};
    
    \end{axis}
    \end{tikzpicture}
    \caption{Threads\,=\,32\label{fig:bench3_2}}
  \end{subfigure}
	\caption{Throughput \vs Payload Size (250k operations)}
\end{figure*}

\subsection{Durable Skiplist Performance}
\label{sec:isolated_evaluation}
We begin by evaluating the performance of \sys's durable skiplist in isolation (without a capacity tier). As well as preventing interference from background operations that propagate data to the capacity tier, this allows us to ensure that all evaluated implementations use the same memory allocators, node structures, key comparators, node height generators, fanout factors, and other parameters that affect performance.

To put \sys's performance in perspective, we compare against the staging tiers of two representative state-of-the-art systems: \prdb\cite{pmemrocksdb} and ListDB~\cite{listdb}. Before presenting our results, we first briefly describe the experimental configuration for each baseline system, and in particular how we address a critical correctness bug that we discovered in ListDB.

\mypar{\prdb setup} Our comparison with \prdb's staging tier uses the following configuration: we enable the \textit{concurrent memtable writes} and \textit{pipelined write} options. These options allow \prdb to parallelise WAL and mem\-table updates across successive write batches, significantly improving its performance. We also configure \prdb to use \textit{key-value separation}, storing payload values in persistent memory in the same manner as the ListDB and \sys staging-tier implementations evaluated in this work.

\mypar{ListDB setup} During our initial investigations into ListDB's staging tier, we discovered two bugs that can cause it to violate the durable linearizability correctness condition. The essence of the first bug is the chain of records in ListDB's WAL may be persisted out-of-order, such that an entry after an invalid record can become visible to clients but is lost during recovery (see Appendix~\ref{appendix:listdb_bug} for details). The second bug arises from a missing flush operation.

To address the first issue, we designed a fix to ListDB's algorithm that preserves its lock-free properties and ensures that the WAL chain can skip over invalid records, thereby ensuring durable linearizability. For the second issue we add the required flush. We call this fixed version ListDB(sync) and it is the version that we use in our benchmark evaluation.

\mypar{Throughput} Figures~\ref{fig:bench_throughput_1}, \ref{fig:bench_throughput_2}, and~\ref{fig:bench_throughput_3} show the insertion throughput into the WAL--memtable stack across varied thread counts
for \througputFirstSet, \througputSecondSet, and \througputThirdSet operations, respectively, under a fixed payload size of \fixedPayloadSizeForTheFirstBenchSet.
\sys consistently outperforms \prdb, which implements the same extended API, by 15\%-25\% across the entire range of thread counts. Moreover, the performance gap widens as the number of threads increases.
Additionally, \sys begins to outperform ListDB(sync) at around 10 threads, achieving a throughput improvement of 4\% to 68\%, with the performance gap continuing to widen as the thread count increases.

In a separate set of experiments, shown in Figures~\ref{fig:bench3_1} and~\ref{fig:bench3_2}, we vary the payload size while executing \fixedNumOpsForTheThirdBenchSet operations using 16 and 32 threads. \sys significantly outperforms the other systems for small payloads, but the gap narrows as payload size increases. By 2048~bytes, performance largely converges, as persistent-memory bandwidth becomes the dominant limiting factor for all 3 systems.
At 32 threads, \sys outperforms \prdb by 18\%--29\% and ListDB by 12\%--69\%, and at 16 threads by 15\%--8\% and 50\%--9\% respectively.

Finally, we evaluate the performance of \Op{WriteBatch} operations for \sys and \prdb, as shown in Figure~\ref{fig:bench_write_batch}. \sys consistently outperforms \prdb by 10\%-15\% across the entire range of batch sizes. Similar to single-update operations, \sys achieves this performance advantage by parallelising the \phaseLocate, \phasePrepare, and \phasePromote phases.
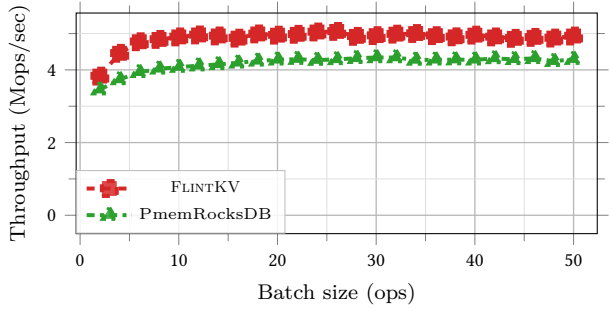
\begin{figure}

\begin{tikzpicture}
\begin{axis}[
    bench style,
    xlabel={Batch size (ops)},
    ylabel={Throughput (Mops/sec)},
    ymin=0.0,
          legend style={
      at={(0,0)},
      anchor=south west,
    }
]

    \addplot[
color=veliriscolor, dashed, mark=square*
    ] coordinates {
        (2,3.8296) (4,4.4565) (6,4.7747) (8,4.838) (10,4.9075) (12,4.9616) (14,4.9342) (16,4.8735) (18,4.9894) (20,4.9554) (22,4.9789) (24,5.043) (26,5.0586) (28,4.9223) (30,4.9428) (32,4.984) (34,5.001) (36,4.9423) (38,4.9101) (40,4.9566) (42,4.9072) (44,4.868) (46,4.8945) (48,4.8758) (50,4.9314)
    };
    \addlegendentry{\sys}

    \addplot[
        color=groupcommitcolor, dash dot, mark=triangle*
    ] coordinates {
        (2,3.4554) (4,3.7366) (6,3.9341) (8,4.0293) (10,4.0767) (12,4.115) (14,4.1501) (16,4.194) (18,4.2482) (20,4.2773) (22,4.2958) (24,4.269) (26,4.2889) (28,4.32) (30,4.3412) (32,4.3402) (34,4.2729) (36,4.2774) (38,4.2819) (40,4.2773) (42,4.3146) (44,4.2905) (46,4.3175) (48,4.2382) (50,4.2945)
    };
    \addlegendentry{\pmemRocksDB}

\end{axis}
\end{tikzpicture}
\caption{WriteBatch: Throughput \vs Batch Size\label{fig:bench_write_batch}}
\end{figure}

\mypar{Latency\label{sec:algo_vacuum_bench_latency}}
Figures~\ref{fig:bench2_p50} and~\ref{fig:bench_2_p99} show the median (P50) and 99th-percentile (P99) latencies, respectively, while Figure~\ref{fig:bench2_tail} shows the tail amplification (P99/P50) for the same experiments.
For this set of experiments we run \througputFirstSet operations with the payload size of \fixedPayloadSizeForTheFirstBenchSet, which corresponds to the same configuration as in Figure~\ref{fig:bench_throughput_1}.
As with the experiment in Figure~\ref{fig:bench_throughput_1}, \sys has lower latency than \prdb across the entire range of thread counts, with the performance advantage widening as the number of threads increases.
The ListDB(sync) implementation exhibits lower latency than \sys at low thread counts, but its latency increases sharply as the number of threads grows, eventually surpassing \sys at around 10 threads. This is consistent with the throughput results, where \sys begins to outperform ListDB(sync) at similar thread counts.

\subsection{End-to-End Performance\label{sec:end_to_end}}
To understand the impact of \sys on end-to-end performance and the engineering effort required to integrate its durable skiplist into existing systems, we next evaluate \sys in conjunction with the capacity tiers of ListDB and \prdb.

\myparr{\sysplain+ListDB.} In its original form, ListDB only supports \PUT and \GET operations. Supporting the full API of \sys requires us to address the lack of support for range scans: because ListDB's data is sharded using a hash-based partitioning scheme, a global range scan would in practice require scanning all shards, which is highly inefficient. However, extending ListDB to support range scans within a single shard is straightforward. We therefore configure ListDB to use a single shard and implement a straightforward range scan API that iterates over both the durable staging tier and capacity tier, returning all keys within the specified bounds.

We use this range-scan-enabled configuration as our baseline ListDB implementation for the integration study.
For the integrated {\sc \sysplain}+ListDB system, we replace ListDB's original staging tier with \sys, while keeping the rest of the shardless ListDB baseline unchanged. 
As a result, the integrated system supports the extended API, which includes \Op{WriteBatch} and \Op{Snapshot}.

\myparr{\sysplain+\prdb.} Unlike ListDB, \prdb already supports \sys's extended API (including \Op{WriteBatch} and \Op{Snapshot}), so integrating with it does not require modifying the system's core codebase. Instead, we simply replace the existing durable staging tier with \sys, leaving the remainder of \prdb untouched.

\myparr{Throughput.} 
Figures \ref{fig:veliris_with_list_db_end_to_end} and \ref{fig:pmemrocksdb_veliris_throughput} show the end-to-end throughput of {\sc \sysplain}+ListDB and {\sc \sysplain}+\prdb, respectively, in comparison to their original versions. For ListDB, performance exhibits a trend similar to that in \cref{sec:isolated_evaluation}, with \sys outperforming ListDB(sync) at higher thread counts by by 8\%--75\%.
We attribute the performance drop when scaling from 6 to 8 threads for both systems to ListDB's memtable rotation mechanism. This relies on reference counting to determine when the active memtable can be safely made immutable. Combined with the synchronisation overhead for allocating DRAM and NVM space for new memtable entries, this temporarily degrades performance.

With respect to \prdb (Figure \ref{fig:pmemrocksdb_veliris_throughput}), \sys's end-to-end throughput also demonstrates performance improvment of 49\%--73\% across the evaluated configurations. 


\begin{figure}[ht]

\begin{tikzpicture}
\begin{axis}[
    bench style,
    xlabel={Number of Threads},
    ylabel={Throughput (ops/sec)},
    ymin=0.0,
    legend style={
      at={(0,0)},
      anchor=south west,
    }
]
\addplot [veliriscolor, dashed, mark=square*] coordinates {
    (2, 1312661)
    (4, 1594568)
    (6, 1795437)
    (8, 1499502)
    (10, 1657758)
    (12, 1764370)
    (14, 1862380)
    (16, 1957828)
    (18, 2027558)
    (20, 2152996)
    (22, 2271899)
    (24, 2332754)
    (26, 2390303)
    (28, 2478479)
};
\addlegendentry{\sys}
\addplot [color=listdbcolor, solid, mark=*] coordinates {
    (2, 1335161)
    (4, 1630018)
    (6, 1655288)
    (8, 1393116)
    (10, 1400562)
    (12, 1384562)
    (14, 1381683)
    (16, 1352883)
    (18, 1445615)
    (20, 1464133)
    (22, 1456814)
    (24, 1422970)
    (26, 1415813)
    (28, 1414263)
};
\addlegendentry{ListDB}

\end{axis}
\end{tikzpicture}
	\caption{End-to-End Throughput: ListDB vs. \sys using \texttt{db\_bench} fillrandom benchmark (1M operations, 128 byte payload).\label{fig:veliris_with_list_db_end_to_end}}
\end{figure}
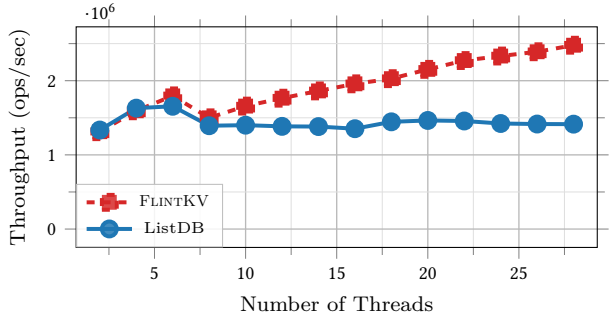

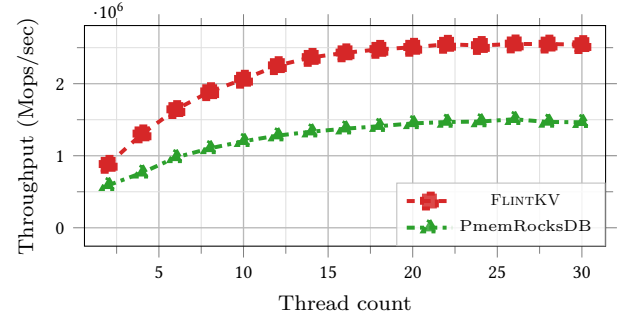
\begin{figure}[ht]
\centering
\begin{tikzpicture}
\begin{axis}[
    bench style,
    xlabel={Thread count},
    ylabel={Throughput (Mops/sec)},
    ymin=0.0,
    legend style={
      at={(1,0)},
      anchor=south east,
    }
]
\addplot[
    color=veliriscolor, dashed, mark=square*
    ]
    coordinates {
    (2, 880864) (4, 1303461) (6, 1640814) (8, 1889337) (10, 2062132) (12, 2250874) (14, 2365221) (16, 2428622) (18, 2476804) (20, 2505436) (22, 2546987) (24, 2530560) (26, 2551909) (28, 2549985) (30, 2542886)
    };
    \addlegendentry{\sys}

\addplot[
color=groupcommitcolor, dash dot, mark=triangle*
    ]
    coordinates {
    (2, 589297) (4, 763093) (6, 977010) (8, 1101797) (10, 1201045) (12, 1279130) (14, 1333809) (16, 1375141) (18, 1409832) (20, 1449052) (22, 1468040) (24, 1474943) (26, 1506354) (28, 1468751) (30, 1464867)
    };
    \addlegendentry{PmemRocksDB}

\end{axis}
\end{tikzpicture}
\vspace{0.5cm}
	\caption{End-to-End Throughput: PmemRocksDB vs. \sys using \texttt{db\_bench} fillrandom benchmark (1M operations, 128 byte payload).}
\label{fig:pmemrocksdb_veliris_throughput}
\end{figure}


\subsection{Recovery Performance\label{sec:recovery_time_evaluaiton}}
We next evaluate the recovery time of \sys, which consists of walking a singly-linked list of log records on \NVMM and inserting each key into a DRAM skip list. For \textit{\sys}, records are scattered at random physical offsets on \NVMM but linked in ascending key order (sorted). We compare against recovery using a standard write-ahead log (WAL) where records are physically contiguous (sequential scan, prefetcher-friendly) but linked in random key order (unsorted).

Each record occupies one 64-byte cache line (8-byte key, 8-byte next pointer, 48-byte padding), allowing 4 records to fit in a single 256B \NVMM page. We vary WAL size from 16 MiB to 1024 MiB and evict all records from cache before each timed run via \texttt{clflush}+\texttt{sfence}, ensuring every access incurs a true \NVMM read.

The results are shown in Figure~\ref{fig:recovery}. The sorted layout of \sys is faster across all sizes, with the advantage peaking at 128 MiB log sizes. At 16 MiB both layouts fit largely in the CPU cache and the gap is small (14.4\%). As the working set grows up to the size of the CPU cache the insertion cost dominates, with the benefits of \sys's sorted insertion peaking at 52.8\% at 128 MiB. At larger sizes the performance advantage of \sys's approach is reduced to to 7.9\% at 1024 MiB. This is due to repeated reads from \NVMM when a record is accessed after a cacheline containing a previously accessed record on the same \NVMM page was evicted. 
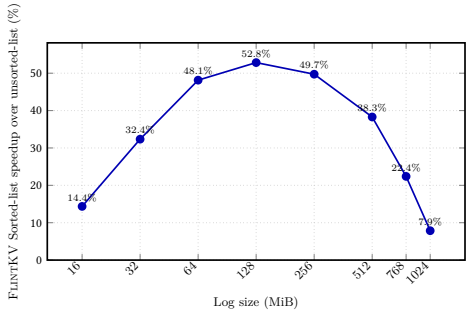
\begin{figure}
    
\centering
\resizebox{0.35\textwidth}{!}{
\begin{tikzpicture}
\begin{axis}[
    width  = 12cm,
    height = 7cm,
    xlabel = {Log size (MiB)},
    ylabel = {\sys Sorted-list speedup over unsorted-list (\%)},
    xmode  = log,
    log basis x = 2,
    xtick  = {16, 32, 64, 128, 256, 512, 768, 1024},
    xticklabels = {16, 32, 64, 128, 256, 512, 768, 1024},
    xticklabel style = {rotate=45, anchor=east},
    ymin   = 0,
    ymajorgrids = true,
    xmajorgrids = true,
    grid style  = {dotted, gray!50},
    mark size   = 2.5pt,
    line width  = 1.2pt,
    every axis plot/.append style = {color=blue!70!black},
]

\addplot[
    color  = blue!70!black,
    mark   = *,
    mark options = {fill=blue!70!black},
] coordinates {
        (16, 14.36)
        (32, 32.35)
        (64, 48.14)
        (128, 52.84)
        (256, 49.73)
        (512, 38.29)
        (768, 22.39)
        (1024, 7.86)
};

\node[above, font=\footnotesize] at (axis cs:16,14.36) {14.4\%};
\node[above, font=\footnotesize] at (axis cs:32,32.35) {32.4\%};
\node[above, font=\footnotesize] at (axis cs:64,48.14) {48.1\%};
\node[above, font=\footnotesize] at (axis cs:128,52.84) {52.8\%};
\node[above, font=\footnotesize] at (axis cs:256,49.73) {49.7\%};
\node[above, font=\footnotesize] at (axis cs:512,38.29) {38.3\%};
\node[above, font=\footnotesize] at (axis cs:768,22.39) {22.4\%};
\node[above, font=\footnotesize] at (axis cs:1024,7.86) {7.9\%};

\end{axis}
\end{tikzpicture}
}
\caption{Recovery performance: sorted vs.\ unsorted log. The sorted layout of \sys is faster across all sizes, with the advantage peaking at 128 MiB log size.}
\label{fig:recovery}
\end{figure}
\section{Related Work}
\label{sec:related-work}

\mypar{NVM KV Stores} ListDB~\cite{listdb} and \prdb~\cite{pmemrocksdb} are hybrid \NVMM KV stores that support multi-versioning, and are therefore a good match for integration with \sys. In contrast, the previously discussed systems in \cref{sec:hybrid-nvmm-dram}—BonsaiKV~\cite{bonsaikv}, FluidKV~\cite{fluidkv}, and Viper~\cite{viper}—do not provide multi-versioning support. Instead, they update entries in place, retaining only the most recent version of each key-value pair in the system.

\mypar{\NVMM Indexes} A large body of work has explored the design of \NVMM indexes, including B+ trees~\cite{fftree,wb+tree,nvtree,lbtree,nap,fptree} and hash-based indexes~\cite{dash}. 
In this work, we focus on skip lists because they are the de facto in-memory index for LSM-tree memtables. Memtables are transient data structures that grow from empty to a fixed size before being flushed to persistent storage and discarded. For this workload, skip lists provide efficient ordered access while avoiding the structural maintenance 
(e.g., rebalancing or node splitting) required by many tree-based indexes, making them a natural choice for write-intensive, short-lived in-memory data structures.

\mypar{\NVMM Skiplists} A number of prior works~\cite{pmem_skiplist,phast_pmem_skiplist,chen2019_pmem_skiplist,another_pmem_skiplist} have explored the design of \NVMM skip lists. However, these efforts primarily focus on optimizing individual operations and do not support atomic write batches or snapshots.
\textsc{NVTraverse}~\cite{DBLP:conf/pldi/FriedmanBWBP20} transforms lock-free data structures for persistent memory and introduces the concept of a \textit{traversal phase}, which inspired the design of \sys. However, like the other prior works discussed above, \textsc{NVTraverse} does not support atomic write batches or snapshots.

\mypar{Lock-free data structures with extended APIs} Recent work on lock-free data structures includes Jiffy~\cite{kobus2022jiffy}, which integrates atomic batch updates and snapshot semantics into a lock-free skiplist, and VERLIB~\cite{blelloch2024verlib}, which enables snapshot-consistent traversal of concurrent data structures via versioned pointers and per-pointer CAS-based updates. Both systems demonstrate high performance in volatile-memory settings, but neither addresses crash consistency or persistence semantics required for non-volatile memory systems.

\mypar{Flat-combining data structures} A body of work has explored flat-combining data structures for both volatile memory and \NVMM. These efforts have primarily focused 
on data structures with little or no traversal, such as stacks, queues, and priority queues~\cite{DBLP:conf/spaa/HendlerIST10,DBLP:conf/sss/RusanovskyABGHR21,fatourou2022performance,egorov2025fast}. 
In contrast, the application of flat-combining to traversal-intensive data structures, such as linked lists, skip lists, and trees, remains largely unexplored.

A possible reason is that flat-combining, in its original form, is less well suited to traversal-intensive data structures. In a naive implementation of a linked list or skiplist based on flat-combining, the combiner thread both traverses the data structure and executes all pending operations. As the data structure grows, this centralized traversal becomes an increasingly significant bottleneck, severely limiting throughput and scalability. In contrast, \sys adopts a phase-based approach that allows multiple threads to do most of the work associated with traversing and preparing operations concurrently, leaving only the minimal coordination necessary to complete the operations.

\section{Conclusions}
\label{sec:conclusions}
In this work, we present \sys, an NVM-optimized storage engine designed to bridge the gap between state-of-the-art \NVMM key-value stores and the API requirements of contemporary database management systems. \sys natively supports atomic multi-key writes and consistent snapshot iteration while also guaranteeing durable linearizability. At the heart of \sys is a novel concurrent durable skiplist index whose concurrency control and persistence mechanisms are carefully co-designed to maximize performance, allowing \sys to outperform prior work by up to 75\% in end-to-end throughput. 



\bibliography{refs}
\appendix
\section{ListDB Crash-Consistency Bug}
\label{appendix:listdb_bug}
In summary, to add an entry to the \emph{durable staging tier}, ListDB first reserves space in the WAL using a bump-pointer allocation. It then writes the entry and the offset of the next entry in the WAL and persists it. Finally, it marks entry as valid, which means that the entry will be restored during recovery. This, effectively, forms a linked chain. During recovery, the algorithm traverses this chain and checks whether each entry is valid; once it encounters an invalid entry, it stops.

We observe that this behaviour can violate durable linearizability because an acknowledged write operation is not guaranteed to be visible after a crash.

In particular, recovery relies on the following invariants:
\begin{myitemize}
    \item \textbf{Reachability:} Every successfully committed (i.e., operation that has returned success and therefore must be durable and visible) WAL record must be reachable from the head of the WAL chain by following \texttt{next} pointers.
    \item \textbf{Stop condition:} If the recovery procedure stops at the first invalid record, then no valid record may appear \emph{after} an invalid one in the traversal order.
\end{myitemize}

The following execution violates these invariants.
\begin{myenumerate}
    \item At time 0, Thread~1 reserves space in the WAL and starts persisting the node's data, but has not yet marked the entry as valid.
    \item At time 1, Thread~2 reserves space in the WAL, persists its entry, marks it as valid, inserts the node into the skiplist, and returns.
    \item At time 2, a crash happens.
\end{myenumerate}

After such an execution, the entry written by Thread~2 may not be reachable during recovery, because Thread~1 has not yet created the link that allows the recovery procedure to reach subsequent entries. As a result, operations that returned successfully before the crash can be lost, violating durable linearizability. To address this issue, we designed a fix to their algorithm that preserves its lock-free property and ensures that the WAL chain can skip over invalid records, thereby enforcing the adherence to durable linearizability. We call this fixed version of ListDB ListDB(sync) and it is the verison that we use in our benchmark evaluation.





\newcommand{\newtext}[1]{{\color{red}#1}}

\section{Proof of Durable Linearizability}
\label{sec:dl-proof}

We model the operations in the \sys\/ API (\S\ref{sec:interface}) as
reads and writes over the sets of key/value pairs. 
Specifically, \Op{Get} that returns
the value $\valu$ for the key $\key$ is represented
as $\oread(\{(\key, \valu)\})$, 
\Op{Snapshot} of the keys $\key_1, \key_2, \dots$
returning the values $\valu_1, \valu_2, \ldots$
is represented as $\oread(\{(\key_1, \valu_1), (\key_2, \valu_2), \ldots\})$,
\Op{Put} of a value $\valu$ into the key $\key$
as $\owrite(\{(\key, \valu)\})$, 
\Op{WriteBatch} of the values $\valu_1,$ $\valu_2, \ldots$
into the keys $\key_1, \key_2, \dots$
as $\owrite(\{(\key_1, \valu_1), (\key_2, \valu_2), \ldots\})$,
and \Op{Delete} of the value associated with $\key$ as
$\owrite(\{(\key, \bot)\})$.

For an arbitrary operation $o$, we let 
$\keys(o) = \{\key_1, \key_2, \ldots\}$ and 
$\valus(o) = \{\valu_1, \valu_2, \ldots\}$ denote
the key and the value components of its 
argument, respectively. 

Our proof follows the dependency graph framework originally introduced by
Adya~\cite{adya99-weak-consis} to capture sufficient conditions for
transaction serializability. We specialize this framework to our setting by
replacing read-only and read-write transactions with $\oread$ and $\owrite$
operations, respectively. To reason about operations left incomplete by
crashes (as required by durable linearizability), we follow Khyzha et
al.~\cite{artyom}, who extend~\cite{adya99-weak-consis} with a
\emph{visibility} predicate to accommodate histories with 
both complete and incomplete (commit-pending) transactions. 
In our setting these correspond to complete and
incomplete operations, respectively.


We first formalise the notion of a dependency graph. Note that for
simplicity, we interchangeably write $x \in p$ and $p(x)$ for a
predicate $p$. Given (possibly infinite) execution $\sigma$ of the algorithm, we let:

\begin{itemize}
	\item $\allrdr(\sigma)$ be the set of
  all $\oread$ operations in $\sigma$; 

  \item $\allwrr(\sigma)$ be the set of all 
  $\owrite$ operations in $\sigma$;

	\item $V(\sigma)$ be the set of all the operations in 
	$\sigma$; and
	
	\item $\rt(\sigma) \subseteq 
	V(\sigma)\times 
	V(\sigma)$ be a binary relation defined 
	as follows: for all $o_1,o_2\in V(\sigma)$, 
	$(o_1,o_2)\in\rt(\sigma)$ if 
	and only if $o_1$ completes before $o_2$ is invoked.
\end{itemize}

\begin{definition}
  \label{dep-graph-def}
  Let $\sigma$ be an execution. A \emph{dependency graph} for $\sigma$
  is a tuple $G=(V(\sigma),\rt(\sigma), \vis, \WR,\WW,\RW)$, where the
  visibility predicate $\vis \subseteq V(\sigma)$ and relations
  $\WR,\WW,\RW\subseteq V(\sigma)\times V(\sigma)$ are such that:
	
  \begin{enumerate}[leftmargin=15pt]
  \item $\vis(o)$ holds for all complete operations and for a subset of 
    incomplete operations; and
  \item \begin{enumerate}[label=(\roman*)]
    \item $\WR = \bigcup_{\key \in \allkeys} \WR_\key$, 
      where 
      $\allkeys$ is the set of all keys used in the execution;
    \item for all $\key \in \allkeys$, if 
      $(\wrr,\rdr)\in\WR_\key$, 
      then $\wrr\in \allwrr(\sigma) \cap \vis$, $\rdr\in 
      \allrdr(\sigma)$, and 
      $(\key, 
      \valu) \in \wrr \cap \rdr$ for some $\valu$;
    \item for all $\key \in \allkeys$ and $\wrr_1,\wrr_2,\rdr$ 
      $\in V(\sigma)$ such that 
      $(\wrr_1,\rdr)\in\WR_\key$ and 
      $(\wrr_2,\rdr)\in\WR_\key$, we have $\wrr_1=\wrr_2$;
    \item for all $\key \in \allkeys$, if $\key \in \keys(\rdr)$ for some $\rdr\in 
      \allrdr(\sigma)$ and there exists no $\wrr\in 
      \allwrr(\sigma)$ such that 
      $(\wrr,\rdr)\in \WR_\key$, then $(\key, \bot) 
      \in \rdr$; 
      and
    \end{enumerate}
  \item \begin{enumerate}[label=(\roman*)]
    \item $\WW = \bigcup_{\key \in \allkeys} \WW_\key$, 
      where $\allkeys$ is the set of all keys used in the execution;
    \item for all $\key \in \allkeys$, $\WW_\key$ is a total 
      order over $\{\wrr \in \allwrr(\sigma)\ |$ $\key \in 
      \keys(\wrr) \wedge \vis(\wrr)\}$; and
    \end{enumerate}
  \item \begin{enumerate}[label=(\roman*)]
    \item $\RW = \bigcup_{\key \in \allkeys} \RW_\key$, 
      where $\allkeys$ is the set of all keys used in the execution;
    \item $\RW_\key = \{(\rdr,\wrr) \mid \exists \wrr'.\ 
      (\wrr',\rdr)\in\WR_\key \wedge 
      (\wrr',\wrr)\in\WW_\key\}\ \cup$\\
      \phantom\ \ \ \ \ \ \ $\{(\rdr,\wrr) \mid r\in 
      \allrdr(\sigma) \wedge 
      \wrr\in 
      \allwrr(\sigma) \cap \vis \wedge {}$ \\
      \phantom\ \ \ \ \ \ \  $\neg\exists \wrr'.\ 
      (\wrr',\rdr)\in\WR_\key 
      \wedge 
      \key \in \keys(\rdr) \cap \keys(\wrr)\}$.
    \end{enumerate}
  \end{enumerate}
\end{definition}


To prove durable linearizability, we rely on the theorem by Khyzha et 
al.~\cite{artyom} below. This theorem generalises the original dependency graph framework
of~\cite{adya99-weak-consis} to accommodate operations that may not complete.


\begin{theorem}
  \label{thm:linearizability-orig}
  An history $\sigma$ is linearizable if there exist $\vis$, $\WR$,
  $\WW$, and $\RW$ such that $G = (V(\sigma), \rt(\sigma), \vis,$
  $\WR, \WW, \RW)$ is an acyclic dependency graph.
\end{theorem}

Recall that a history $\sigma$ is durably linearizable iff (i) any
thread executing before a crash does not resume after the crash and
(ii) the history $\sigma$ with crashes removed is linearizable. This
close relationship between linearizability and durable linearizability
means that \cref{thm:linearizability-orig} can be readily extended to
durable linearizability.
\begin{theorem}
  \label{thm:linearizability}
  An history $\sigma$ is durably linearizable if there exist $\vis$,
  $\WR$, $\WW$, and $\RW$ such that
  $G = (V(\sigma), \rt(\sigma), \vis,$ $\WR, \WW, \RW)$ is an acyclic
  dependency graph.
\end{theorem}
Here, the $\vis$ predicate is used to indicate whether an operation
has taken effect; if an operation does not return due to a crash, it
is not known to they system whether it was successful.

We now prove that every history of \sys is durably linearizable.
Fix one such history $\sigma$. Our strategy is to find witnesses for 
all $\WR_\key$ and $\WW_\key$ that validate the conditions of 
Theorem~\ref{thm:linearizability}. To this end, consider the following 
definition.

\begin{definition}\label{def:tau-tag}

Function $\tau:\sigma\rightarrow \mathbb{N} \cup \{\bot\}$
maps each operation $o \in V(\sigma)$ to a version as follows:
\begin{itemize}
	\item 
	For a $\oread$ operation $\rdr$, $\tau(\rdr)$ equals the value of\\
	$\visibleSeqt$ read in line~\ref{code:get_op_sid_cache} of \cref{algo:get_operation} if $\rdr$ is 
	a \Op{Get}. If $\rdr$ is a \Op{Snapshot}, $\tau(\rdr)$ equals the value of
	$\visibleSeqt$ read at the beginning of the \Op{Snapshot} operation as described in 
	\cref{sec:detailed_algo_finish}.
	
	\item 
	For a $\owrite$ operation $\wrr$, $\tau(\wrr)$ equals the value written to
	$\visibleSeqt$ in line~\ref{code:increment_seq_id} of \cref{lab:insert_critical} 
	if $\wrr$ is 
	a \Op{Put}. If $\wrr$ is a \Op{WriteBatch}, $\tau(\wrr)$ equals the value written to 
	$\visibleSeqt$ in line~\ref{code:batch_seq_id_assignment} of 
	\cref{algo:process_wb}.
	
	\item If an operation $o$ does not execute the corresponding line 
	specified above, $\tau(o) = \bot$.
\end{itemize}

\end{definition}

We now define the witnesses as follows:

\begin{itemize}
	\item $\vis(\rdr)$ holds for $\rdr \in \allrdr(\sigma)$ iff $\rdr$ 
	completes and returns;
	
	\item $\vis(\wrr)$ holds for $\wrr \in \allwrr(\sigma)$ iff $\wrr = \text{\Op{Put}}$ and $\wrr$ 
	persists line~\ref{code:pmem_node_insert} in \cref{lab:insert_critical} to persistent memory or 
	$\wrr = \text{\Op{WriteBatch}}$ and $\wrr$
	persists line~\ref{code:commit_batch} in \cref{algo:process_wb} to persistent 
	memory.

      \item $(\wrr,\rdr)\in\WR_\key$ if and only if

        \begin{itemize}
        \item $\wrr \in 
          \allwrr(\sigma) \cap \vis$
        \item 
          $\rdr \in \allrdr(\sigma)$
        \item $\tau(\wrr) \le 
          \tau(\rdr)$
        \item  
          $\key \in 
          \keys(\wrr) \cap \keys(\rdr)$ and
        \item 
          $\neg\exists \wrr'.$ $\vis(\wrr') \wedge \key \in 
          \keys(\wrr') \wedge \tau(\wrr) < 
          \tau(\wrr') \le 
          \tau(\rdr)$
        \end{itemize}
	
      \item $(\wrr,{\wrr}')\in\WW_\key$ if and only if

        \begin{itemize}
        \item $\wrr,\wrr' \in 
	\allwrr(\sigma) \cap \vis$
      \item  
	$\tau(\wrr)<\tau({\wrr}')$ and 
      \item $\key \in 
	\keys(\wrr) 
	\cap \keys(\wrr')$
        \end{itemize}
 and
	
	\item $\RW_\key$ is derived from $\WR_\key$ and $\WW_\key$ as per the 
	dependency 
	graph definition.
	
\end{itemize}

To show that our witnesses satisfy the requirements of the dependency
graph, we rely on the following properties. It is easy to see that the
algorithms described in \cref{sec:updates,sec:recovery} satisfy the following lemma:

\begin{lemma}
	\label{abd:tags_lemma}
	\
	\begin{enumerate}
        \item For any $\wrr_1,\wrr_2\in \allwrr(\sigma) \cap \vis$
          if $\tau(\wrr_1)=\tau(\wrr_2)$ then
          $\wrr_1=\wrr_2$.\label{abd:tags_lemma1}
		
		
		\item For every $\wrr, \rdr \in V(\sigma)$ and $\key 
                  \in \allkeys$, if $(\wrr, \rdr) \in \WR_\key$ then 
                  $\exists \valu.\ (\key, 
		\valu) \in \wrr \cap \rdr$.

                \item For every $\rdr \in \allrdr(\sigma)$
                  and 
                  $\key \in \keys(\rdr)$ if $\neg \exists \wrr'.\
                  (\wrr', \rdr) \in \WR_\key$ then $(\key, \bot) \in
                  \wrr \cap \rdr$.

	\end{enumerate}
\end{lemma}



Our proof also relies on the following auxiliary lemma:
\balance
\begin{lemma}
	\label{abd:tau-non-inc}
	\
	\begin{enumerate}
		\item For all $\key$ and $\rdr,\wrr\in V(\sigma)$, if 
		$(\rdr,\wrr)\in\RW_\key$ then 
		$\tau(\rdr)<\tau(\wrr)$.
		\item For all $o_1,o_2\in V(\sigma)$, if 
		$(o_1,o_2)\in\RT$ and $\vis(o_1) \wedge \vis(o_2)$, 
		then 
		$\tau(o_1)\leq\tau(o_2)$.
		Moreover, if $o_2$ is a $\owrite$, then 
		$\tau(o_1)<\tau(o_2)$.\label{abd:tau-non-inc2}
	\end{enumerate}
\end{lemma}

\begin{proof}
	\
	\begin{enumerate}
		\item Let $\key$ and $\rdr,\wrr\in V(\sigma)$ be such that 
		$(\rdr,\wrr)\in\RW_\key$. 
		There 
		are two cases.
		\begin{itemize}
			\item Suppose that for some
			$\wrr'$ we have $(\wrr',r)\in\WR_\key$. It must be the 
			case 
			that 
			$({\wrr}',\wrr)\in\WW_\key$. 
			$(\wrr',r)\in\WR_\key$ 
			implies that 
			$\tau(\wrr') \le \tau(\rdr)$, and 
			$({\wrr}',\wrr)\in\WW_\key$ 
			implies $\tau(\wrr') < \tau(\wrr)$ and $\vis(\wrr)$. If 
			$\tau(\wrr) 
			\le 
			\tau(\rdr)$, then according to our definition of 
			$\WR_\key$, 
			it 
			must be the case that $(\wrr', 
			\rdr) \notin \WR_\key$. A contradiction.
			\item Suppose now that $\neg \exists {\wrr}'.\ 
			({\wrr}',\rdr)\in\WR_\key$. If $\tau(\wrr) 
			\le 
			\tau(\rdr)$, then 
			according to our definition of $\WR_\key$, it must be the case that 
			$\exists \wrr'.\ (\wrr', \rdr) \in \WR_\key$. A 
			contradiction.
		\end{itemize}
			\item Let $o_1$, $o_2 \in V(\sigma)$ be such that 
			$(o_1, 
			o_2) \in 
			\RT(\sigma)$ and $\vis(o_1) \wedge \vis(o_2)$. Since $o_1$ 
			completes, Let  $u$ be 
			the value read by $o_2$ 
			from 
			$\visibleSeqt$ in the following lines: 
			line~\ref{code:get_op_sid_cache} in \cref{algo:get_operation}
			if $o_2=\ $\Op{Get}; line~\ref{code:index_seq_read} in \cref{lab:insert_critical}
			if $o_2=\ $\Op{Put}; and line~\ref{code:start_batch} in \cref{algo:process_wb} 
			if $o_2=\ $\Op{WriteBatch}. 
			If $o_2=\ $\Op{Snapshot}, then let $u$ be the value of
			$\visibleSeqt$ read at the beginning of the \Op{Snapshot} operation as described in 
			Section~\ref{sec:detailed_algo_finish}.
			Since $o_2$ starts after $o_1$ completes, operation $o_2$ reads $u$ after
			$o_1$ reads or writes $\tau(o_1)$ from or to 
			$\visibleSeqt$. 
			Because $\vis(o_1)$ holds, we must have $u \ge 
			\tau(o_1)$. We 
			now have one of the following:
			\begin{itemize}
				\item If $o_2$ is a $\oread$, $u$ is used as 
				$\tau(o_2)$ 
				per 
				our definition, and we have $\tau(o_2) = u \ge 
				\tau(o_1)$.
				
				\item If $o_2$ is a $\owrite$, it increments $u$ before writing it
				to $\visibleSeqt$ (line~\ref{code:index_seq_read} in \cref{lab:insert_critical} or 
				line~\ref{code:batch_seq_id_assignment} in \cref{algo:process_wb}), 
				and thus 
				$\tau(o_2) > 
				u \ge \tau(o_1)$.
			\end{itemize}
			This proves the statement.
	\end{enumerate}
\end{proof}

\begin{theorem}
	\label{thm:dep-graph-acyc}
	$G=(V(\sigma), \RT(\sigma), \vis, \WR, \WW, 
	\RW)$ is an acyclic 
	dependency 
	graph.
\end{theorem}

\begin{proof}
	From Lemma~\ref{abd:tags_lemma} 
        and the 
	definitions of $\vis$,
	$\WR$, $\WW$ and $\RW$ 
	it follows that $G$ is a dependency graph. 
	We now show that $G$ is acyclic. By contradiction, assume that the graph 
	$G$ contains a 
	cycle $o_1, \dots, o_n = o_1$ where $\vis$ holds for all $o_1, \dots, 
	o_n = o_1$. Then $n > 1$. 
	By Lemma~\ref{abd:tau-non-inc} and the definitions of $\tau$, $\WW$ and 
	$\WR$, we must have 
	$\tau(o_1) \leq \dots \leq \tau(o_n) = 
	\tau(o_1)$, so that 
	$\tau(o_1) = 
	\dots = \tau(o_n)$. 
	Furthermore, if $(o,o')$ is an edge of $G$ and $o'$ is a $\owrite$, 
	then 
	$\tau(o) < \tau(o')$. 
	Hence, all the operations in the cycle must be $\oread$s, and thus, all the 
	edges in the 
	cycle come from $\rt(\sigma)$. Then there exist $\oread$s $r_1$, 
	$r_2$ in 
	the 
	cycle such that 
	$r_1$ completes before $r_2$ is invoked and $r_2$ completes before $r_1$ is 
	invoked, 
	which is a contradiction.
\end{proof}

\balance

\end{document}